\documentclass[11pt]{article}
\usepackage{amsfonts,amsmath}
\usepackage{hyperref}
\usepackage{latexsym}
\usepackage{amsthm}
\usepackage{amscd}
\usepackage{epsfig}
\usepackage{tikz}
\usepackage{tikz-cd}
\usepackage{mathtools}
\usepackage{braids}
\usepackage{graphicx}
\usepackage{caption}
\usepackage{enumerate}
\usepackage[shortlabels]{enumitem}
\usepackage{subcaption}
\usetikzlibrary{decorations.pathreplacing}
\usepackage{comment}

\usepackage{float}
\usepackage{soul}
\usepackage{bm}
\usepackage{YBR_functions}
\usepackage{dsfont}

\numberwithin{equation}{section}

\newtheorem{thm}{Theorem}[section]
\newtheorem{prop}[thm]{Proposition}
\newtheorem{lem}[thm]{Lemma}
\newtheorem{cor}[thm]{Corollary}

\newcommand{\be}{\begin{equation}}
\newcommand{\ee}{\end{equation}}
\newcommand{\bib}{\bibitem}
\newcommand{\nc}{\newcommand}

\nc{\Cbb}{\mathbb{C}}
\nc{\Nbb}{\mathbb{N}}
\nc{\Rbb}{\mathbb{R}}
\nc{\Zbb}{\mathbb{Z}}
\nc{\Ph}{\hat{P}}
\nc{\Qh}{\hat{Q}}
\nc{\sh}{\hat{s}}
\nc{\deh}{\hat{\delta}}
\nc{\ip}{\mathcal{P}}
\nc{\Ac}{\mathcal{A}}
\nc{\Bc}{\mathcal{B}}
\nc{\Dc}{\mathcal{D}}
\nc{\Ec}{\mathcal{E}}
\nc{\Fc}{\mathcal{F}}
\nc{\Lc}{\mathcal{L}}
\nc{\Oc}{\mathcal{O}}
\nc{\Sc}{\mathcal{S}}
\nc{\Vc}{\mathcal{V}}
\nc{\Zc}{\mathcal{Z}}
\nc{\Pb}{\mathsf{P}}
\nc{\tr}{\mathrm{tr}}
\nc{\gh}{\hat{g}}
\nc{\id}{\mathrm{id}}
\nc{\wj}{\mathsf{w}}
\nc{\Fo}{\overline{F}}
\nc{\ko}{\overline{k}}
\nc{\Ko}{\overline{K}}
\nc{\wo}{\overline{w}}
\nc{\Wo}{\overline{W}}
\nc{\xo}{\overline{x}}
\nc{\Xo}{\overline{X}}
\nc{\yo}{\overline{y}}
\nc{\Yo}{\overline{Y}}
\nc{\at}{\tilde{a}}
\nc{\Tt}{\tilde{T}}
\nc{\io}{\overline{i}}
\nc{\ioa}{\overline{1}}
\nc{\iob}{\overline{2}}
\nc{\ioc}{\overline{3}}
\nc{\iod}{\overline{4}}

\definecolor{dblue}{rgb}{.61,.61,1}
\definecolor{lightblue}{rgb}{.61,.61,1}
\definecolor{altblue}{rgb}{.61,.61,1}

\begin{document}

\thispagestyle{empty}

\begin{center}

\scalebox{1.0}{\textbf{\LARGE Integrable models from singly generated planar algebras}}
\\[0.6cm]
{\LARGE Xavier Poncini and J{\o}rgen Rasmussen}
\\[0.3cm]
\textit{School of Mathematics and Physics, University of Queensland\\ 
St Lucia, Brisbane, Queensland 4072, Australia}
\\[0.3cm] 
\textsf{x.poncini\!\;@\!\;uq.edu.au\quad\ j.rasmussen\!\;@\!\;uq.edu.au}

\vspace{1.4cm}

{\large\textbf{Abstract}}\end{center}
Not all planar algebras can encode the algebraic structure of a Yang--Baxter integrable model 
described in terms of a so-called homogeneous transfer operator. 
In the family of subfactor planar algebras, we focus on the ones known as singly generated
and find that the only such planar algebras underlying homogeneous Yang--Baxter integrable models are
the so-called Yang--Baxter relation planar algebras. According to a result of Liu,
there are three such planar algebras: the well-known Fuss--Catalan and Birman--Wenzl--Murakami
planar algebras, in addition to one more which we refer to as the Liu planar algebra.
The Fuss--Catalan and Birman--Wenzl--Murakami algebras are known to underlie Yang--Baxter integrable models,
and we show that the Liu algebra likewise admits a Baxterisation. We also show that 
the homogeneous transfer operator describing a model underlied by a singly generated 
Yang--Baxter relation planar algebra is polynomialisable, 
meaning that it is polynomial in a spectral-parameter-independent element of the algebra.

\newpage

\tableofcontents

\newpage

\section{Introduction}
\label{Sec:Intro}

In \cite{Jones17}, Jones proposed the ``royal road" as a route to construct a conformal field theory 
from a so-called finite index subfactor, where one extracts, from the subfactor,
Boltzmann weights of a critical two-dimensional lattice model, and uses these to construct 
a quantum field theory in the scaling limit. Despite the simplicity and promise of this proposal, Jones stressed 
himself that the available mathematical tools are likely insufficient in all but the simplest examples. The present 
work is an effort to expand our toolkit such that new and more scenic routes may one day be passable.

To set the stage, we recall that a \textit{subfactor planar algebra} \cite{Jones99,Jones00,KS04,JonesNotes} 
is a shaded planar algebra, $(A_{n,\pm})_{n\in\Nbb_0}$, that has a certain set of simple and physically 
well-motivated properties, including an inner product on each $A_{n,\pm}$. A subfactor planar algebra is said to be 
\textit{singly generated} \cite{BJ00,BJ03,BJL17} if it is generated by only \textit{one} generator aside from the 
default Temperley--Lieb generators \cite{TL71,Jones83}, and it is called a \textit{Yang--Baxter relation} 
(YBR) \textit{planar algebra} \cite{Liu15} if it satisfies certain algebraic relations generalising the familiar braid 
relations. According to Liu \cite{Liu15}, there are three singly generated YBR planar algebras (as well as 
quotients thereof): the two-colour Fuss--Catalan (FC) \cite{BJ97} and the Birman--Wenzl--Murakami (BMW) 
\cite{Mur87,BW89} planar algebras, in addition to one more which we refer to as the \textit{Liu planar algebra}.

Relations of the YBR form are typically needed to construct spectral-parameter dependent solutions to the 
Yang--Baxter equations \cite{McGuire64,Yang67,Baxter71,BaxterBook}, through a procedure called Baxterisation 
\cite{Jones90}. Indeed, the FC and BMW planar algebras are known to underlie Yang--Baxter integrable models 
\cite{DiF98,CGX91}, while, to the best of our knowledge, no integrability observation has been presented in the 
case of the Liu planar algebra. We remedy this by Baxterising the Liu algebra within the algebraic integrability 
framework put forward by the authors in \cite{PR22}.

Our main result is a classification of singly generated planar algebras that can \textit{encode} the algebraic 
structure of a \textit{homogeneous Yang--Baxter integrable model}.
Here, up to technical details, a planar algebra is said to \textit{encode} the homogeneous Yang--Baxter integrability 
if no proper planar subalgebra can take its place (the precise definition is given in Section \ref{Sec:YBpol}),
while a \textit{homogeneous Yang--Baxter integrable model} is one that is described by 
a \textit{homogeneous transfer operator}.
Such a transfer operator is built in a particular `uniform' way using a single $R$-operator and a pair of 
$K$-operators, each parameterised using a single parameter $u\in\mathbb{C}$ (see \eqref{Tnu2}) such that a 
set of local relations, including \textit{generalised} Yang--Baxter equations
(see Proposition \ref{Prop:YBE}), are satisfied. Moreover, a transfer operator is said to be 
\textit{polynomialisable} \cite{PR22} if it is polynomial in a spectral-parameter-independent algebra element.
Our findings may now be summarised as follows.

\begin{thm}\label{thm:1}
    Let $A$ be a singly generated planar algebra. 
    Then, there exists a homogeneous Yang--Baxter integrable model with $A$ encoding the integrability
    if and only if $A$ is a Yang--Baxter relation planar algebra.
\end{thm}
\begin{thm}\label{cor:1}
	Let the homogeneous Yang--Baxter integrability of a model be encoded by a singly generated 
	Yang--Baxter relation planar algebra.
	Then, there exist algebra-parameter choices and a suitable $u$-domain such that the corresponding 
	transfer operator is polynomialisable.
\end{thm}

As a unifying framework for describing the singly generated YBR planar algebras, inspired by the series of 
works \cite{BJ00, BJ03, BJL17, Liu15}, we find it convenient to introduce the unshaded
\textit{proto-singly-generated} (PSG) \textit{planar algebra} $(\mathrm{PS}_n)_{n\in\Nbb_0}$. Although
$\mathrm{PS}_n$ is infinite-dimensional for $n\geq3$ and does not encode a homogeneous Yang--Baxter 
integrable model, we demonstrate that the PSG planar algebra serves as an `ambient' algebra admitting 
quotients isomorphic to the FC, BMW, and Liu planar algebras. 

The paper is organised as follows.
In Section \ref{Sec:Alg}, we recall the basics of shaded planar algebras, including the subfactor, 
YBR, and Temperley--Lieb variants. Certain technical details can be found in Appendix A and Appendix B.

In Section \ref{Sec:YBIntYBR}, we introduce the PSG planar algebra $(\mathrm{PS}_n)_{n\in\Nbb_0}$. 
For each $n\in\Nbb$, we give a presentation of the corresponding algebra $\mathrm{PS}_n$. 

In Section \ref{Sec:YBim}, we review the algebraic integrability framework developed recently in \cite{PR22}.
Using \cite{LMP17}, we demonstrate that it suffices to consider 
unshaded planar algebras when addressing the integrability questions of our interest.
We then show that the PSG planar algebra cannot encode Yang--Baxter integrability unless additional 
conditions are imposed on its generators such that the corresponding planar algebra is YBR. 

In Section \ref{Sec:FC}, we review the FC planar algebra and the corresponding FC algebras $\mathrm{FC}_n$.
In the two-colour case, with the two colours having the same loop weight, the FC planar algebra is YBR, and we 
show that $\mathrm{FC}_n$ is isomorphic to a quotient of $\mathrm{PS}_n$ with parameters given in terms of the 
FC loop weight. We also recast the known Baxterisation \cite{DiF98} in our notation and find that it is the only 
nontrivial Baxterisation supported by our Yang--Baxter integrability framework.

In Section \ref{Sec:BMW}, we review the BMW planar algebra and the corresponding algebras 
$\mathrm{BMW}_n$. We show that $\mathrm{BMW}_n$ is isomorphic to a quotient of 
$\mathrm{PS}_n$ with parameters given in terms of the BMW parameters. We also recast
the known Baxterisation \cite{CGX91} in our notation and find that it is the only nontrivial Baxterisation 
supported by our Yang--Baxter framework. A description of $\mathrm{BMW}_n$ as a quotient of the
braid-semigroup algebra on $n$ strands is presented in Appendix \ref{app:BS}.

In Section \ref{Sec:Liu}, we review the planar algebra introduced by Liu in \cite{Liu15}, here denoted by 
$\mathrm{L}$. For each $n\in\Nbb$, we introduce the corresponding Liu algebra $\mathrm{L}_n$ 
and show that it is isomorphic to a quotient of $\mathrm{PS}_n$ with parameters given in terms of the loop weight 
appearing in $\mathrm{L}$. We also present the Baxterisation of $\mathrm{L}$ within our integrability 
framework \cite{PR22}. Some technical details are omitted but can be found in \cite{Poncini23}.

In Section \ref{Sec:pol}, we recall the notion of polynomial integrability and show that the three singly generated 
YBR planar algebras (FC, BMW, and Liu) all encode Yang--Baxter integrable models described in terms of 
polynomialisable transfer operators. 

Section \ref{Sec:Discussion} contains some concluding remarks.

Throughout, we let $\mathrm{i}$ denote the imaginary unit, $\Nbb$ the set of positive integers, 
and $\Nbb_0$ the set of nonnegative integers.

\section{Planar algebras}
\label{Sec:Alg}

Here, we review the basics of planar algebras \cite{Jones99,Jones00,JonesNotes}, 
including the definition of subfactor planar algebras and the subclass of Yang--Baxter relation planar algebras. 
We also present the Temperley--Lieb planar algebra \cite{Jones99} and the corresponding subfactor planar 
algebra, and recall its ubiquitous role. The usual Temperley--Lieb algebra \cite{TL71,Jones83}, and 
how it arises from the corresponding planar algebra, is discussed in Appendix \ref{app:TL}.

\subsection{Shaded planar algebras}
\label{subsec:YBRPA}

Informally, a complex \textit{shaded planar algebra} \cite{Jones99,Jones00,JonesNotes} is a collection of 
complex vector spaces $(A_{n,\pm})_{n\in\Nbb_0}$ where vectors can be `multiplied planarly' to form vectors,
and where $(A_{n,\pm})_{n\in\Nbb_0}$ is a common abbreviation for 
$(A_{n,\varepsilon})_{n\in\Nbb_0}^{\varepsilon\in\{+,-\}}$. 
A basis for $A_{n,\pm}$ consists of disks with $2n$ \textit{nodes} (connection points) on their boundary, 
whereby a boundary is composed of nodes and boundary intervals, and the boundary intervals are labelled 
alternatingly by $+$ or $-$. The specification of
the internal structure of the basis disks is a key part of the definition of any given planar algebra.
When disks are combined (`multiplied'), every node is connected to a single other node 
(possibly on the same disk) via non-intersecting strings, and \textit{shaded planar tangles} are the diagrammatic 
objects, defined up to ambient isotopy, that facilitate such combinations. The shading of a planar tangle forms 
a `checkerboard' pattern that `matches' the $+/-$ labelling of the boundary intervals.

In general, a planar tangle has the following features:
\begin{itemize}
    \item An exterior disk called the \textit{output} disk.
    \item A finite set of interior disks called \textit{input} disks.
    \item A finite set of non-intersecting strings connecting the nodes of disks.
    \item For each disk, a marked boundary interval.
    \item A checkerboard shading of the region outside the input disks.
\end{itemize}
An example of a planar tangle is given by
\begin{align}\label{equ:TtangleExample}
    \raisebox{-1.45cm}{\RotPlanarNKTangleExampleAltCNewULShaded}
\end{align}
We denote the output disk of the planar tangle $T$ by $D_0^T$ and the set of input disks by $\Dc_T$.
The number of nodes on the (exterior) boundary of a disk $D$ is denoted by $\eta(D)$ and is necessarily 
even, while the \textit{shading} of $D$, denoted by $\zeta(D)$, is $+$, respectively $-$, if the (exterior) 
marked boundary interval borders an unshaded, respectively shaded, region. 
The marked boundary intervals on the disks disambiguate the alignment of the 
input and output disks and are indicated by red rectangles, see (\ref{equ:TtangleExample}). 

Planar `multiplication' is induced by the action of the planar tangles as multilinear maps.
For a planar tangle $T$, this is denoted by 
\begin{align}
 \Pb_T\colon\bigtimes_{D\in\Dc_T}A_{\eta(D)/2,\, \zeta(D)}\to A_{\eta(D_0^T)/2,\, \zeta(D_0^T)}.
\label{PT}
\end{align}
Pictorially, $\Pb_T$ acts by filling in each of the interior disks $D\in\Dc_T$ with an element of the corresponding 
vector space $A_{\eta(D)/2,\,\zeta(D)}$, in such a way that the nodes match up, the marked intervals are aligned, 
and the shading is consistent. The details of the map $\Pb_T$ specify how one should remove the input disks in 
the picture and identify the image as a vector in $A_{\eta(D_0^T)/2,\,\zeta(D_0^T)}$. If $\Dc_T=\emptyset$, we 
simply write $\Pb_T()$ for the image under $\Pb_T$, and stress that $\Pb_T()$ is distinct from $T$.
Taking $T$ as in \eqref{equ:TtangleExample}, we present the example
\begin{align}
    T= \raisebox{-1.425cm}{\RotPlanarNKTangleExampleAltCNew}, \qquad\quad 
    \Pb_T(v_1,v_2,v_3)
    =  \raisebox{-1.425cm}{\RotPlanarNKTangleExampleAltCGenNew} \in A_{4,+},
    \label{equ:GenPlanarActionExample}
\end{align}
where $v_1\in A_{1,+}$, $v_2\in A_{2,-}$ and $v_3\in A_{3,+}$. Note that we have not 
specified any details about $(A_{n,\pm})_{n\in\Nbb_0}$ or about the action of the planar tangles. 
\medskip

\noindent
\textbf{Remark.}
Unlike in the picture of $T$ in (\ref{equ:GenPlanarActionExample}), disks in $\Dc_T$ are not labelled; however,
to apply the ordered-list notation for the vectors in $\Pb_T(v_1,v_2,v_3)$, it is convenient to label the disks
accordingly. Once drawn as in the second picture in (\ref{equ:GenPlanarActionExample}), no labelling is needed.
\medskip

Planar tangles can be `composed', and consistency between this composition and the associated 
multilinear maps is often referred to as {\em naturality}, see Appendix \ref{app:Nat} for details. 
By specifying the vector spaces $(A_{n,\pm})_{n\in\Nbb_0}$ and the action of planar tangles as 
multilinear maps \eqref{PT} such that naturality is satisfied, one arrives at a \textit{shaded planar algebra}. 
For each $n\in\Nbb_0$, the \textit{identity tangles} are introduced as
\begin{align}
    \mathrm{id}_{n,+} :=
    \raisebox{-0.7275cm}{\RotationPlanarZeroShadedPosSmall}\;,\qquad\quad
    \mathrm{id}_{n,-} :=\! \raisebox{-0.68cm}{\RotationPlanarZeroShadedNegSmall}\;, \qquad\quad
    \Pb_{\mathrm{id}_{n,\pm}}: A_{n,\pm}\to A_{n,\pm},
\label{idnpm}
\end{align}
each having $2n$ spokes. Likewise, the elementary \textit{rotation tangles} are introduced as
\begin{align}
    r_{n,1,+} :=  \!\raisebox{-0.68cm}{\RotPosPos}\;,\quad\;\;
    r_{n,1,-} := \raisebox{-0.7275cm}{\RotPosNeg}\;, \quad\;\;
    r_{n,-1,+} := \!\raisebox{-0.68cm}{\RotNegPos}\;,\quad\;\;
    r_{n,-1,-} := \raisebox{-0.7275cm}{\RotNegNeg}\;,
\label{equ:ElRot}
\end{align}
with corresponding linear maps $\Pb_{r_{n,1,\pm}}: A_{n,\pm}\to A_{n,\mp}$ 
and $\Pb_{r_{n,-1,\pm}}: A_{n,\pm}\to A_{n,\mp}$.
If $A_{n,\pm}$ has no null vectors, then $\Pb_{\mathrm{id}_{n,\pm}}$ acts as the identity operator.
Here, we say that a nonzero $v\in A_{n,\pm}$ is a \textit{null vector} if $\Pb_T(v)=0$ for every planar tangle $T$
for which $\Pb_T$ has domain $A_{n,\pm}$. Moreover, if neither $A_{n,+}$ nor $A_{n,-}$ has any null vectors, 
then $\Pb_{r_{n,-\ell,\mp}}\circ\Pb_{r_{n,\ell,\pm}}=\Pb_{\mathrm{id}_{n,\pm}}$ 
for each $\ell\in\{-1,1\}$. Details are presented in Appendix \ref{app:Unit}. 
\medskip

\noindent
\textbf{Remark.}
A \textit{zero planar algebra} \cite{JonesNotes}, where the vector spaces $(A_{n,\pm})_{n\in\Nbb_0}$ are arbitrary 
and all planar tangles act as the zero map, exclusively contains null vectors. We note 
(i) that each vector space of a planar algebra can be extended to include arbitrarily many null vectors, and 
(ii) for each planar algebra with null vectors, except for a zero planar algebra, 
there exists a corresponding planar algebra without null vectors (obtained by omitting them).
\medskip

A shaded planar algebra $(A_{n,\pm})_{n\in\Nbb_0}$ is not an algebra in the usual sense; 
however, it contains a countable number of (standard) algebras. Indeed, for each $n\in\Nbb_0$, the planar tangles
\begin{align}\label{equ:MultAndCoMultPlanar}
    M_{n,+}:=\raisebox{-1.43cm}{\includegraphics{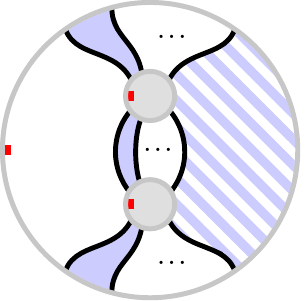}}, \quad\ 
    M_{n,-}:=\raisebox{-1.43cm}{\includegraphics{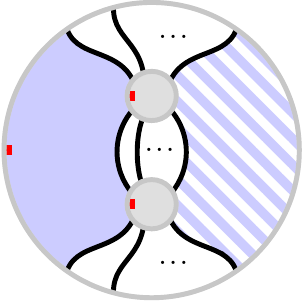}}, \quad\ 
    \Pb_{M_{n,\pm}}: A_{n,\pm}\times A_{n,\pm}\to A_{n,\pm},
\end{align}
induce a \textit{multiplication} on $A_{n,+}$ and $A_{n,-}$, respectively, and we write 
$vw=\Pb_{M_{n,\pm}}(v,w)\in A_{n,\pm}$ for $v,w\in A_{n,\pm}$, where $v$, respectively $w$, 
is replacing the lower, respectively upper, disk in $M_{n,\pm}$. 
\medskip

\noindent
\textbf{Remark.}
If the shading of a region is unspecified, as it may depend on the parity of $n$, 
we use the `banded' pattern illustrated in \eqref{equ:MultAndCoMultPlanar}. 
For simplicity, the region containing the `dots' is coloured white.
\medskip

\noindent
The ensuing algebra, $A_{n,\pm}$, is associative and, if it has no null vectors, unital, with unit
\be
    \mathds{1}_{n,+}:= \Pb_{\mathrm{Id}_{n,+}}(), \qquad
    \mathds{1}_{n,-}:= \Pb_{\mathrm{Id}_{n,-}}(), \qquad
    \mathrm{Id}_{n,+} := \raisebox{-0.675cm}{\includegraphics{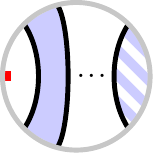}}, \qquad
    \mathrm{Id}_{n,-} := \raisebox{-0.675cm}{\includegraphics{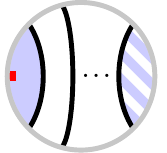}},
\label{Aunit}
\ee
whose dependence on $n$ may be ignored by writing $\mathds{1}_{\pm}$ for $\mathds{1}_{n,\pm}$. 
For details, see Appendix \ref{app:Unit}.

\subsection{Subfactor planar algebras}
\label{subsec:SFPA}

Special classes of planar algebras follow when additional structure is imposed on the vector spaces $A_{n,\pm}$
or on the multilinear maps $\Pb_T$. By doing so, we will encounter \textit{subfactor planar algebras} in the 
following.

In general, there are no constraints on the dimensions of the vector spaces $A_{n,\pm}$, but the 
planar algebra is called \textit{evaluable} if, for each shading $+/-$, $\mathrm{dim}(A_{0,\pm}) = 1$ 
and $\mathrm{dim}(A_{n,\pm}) < \infty$ for all $n\in\Nbb$. In that case, the \textit{evaluation map}
\be
 \mathrm{e}:A_{0,\pm}\to\Cbb,
\label{eval}
\ee
mapping the `empty disk' to the scalar $1$, provides an isomorphism, $A_{0,\pm}\cong\Cbb$.

For each $n\in\Nbb$, the planar tangles
\begin{align}\label{equ:TracePlanar}
        \tr_{n,+}^{(l)}\!:=\!\raisebox{-1.18cm}{\includegraphics{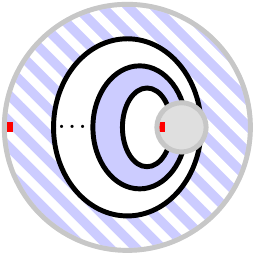}}, &&
        \tr_{n,-}^{(l)}\!:=\!\raisebox{-1.18cm}{\includegraphics{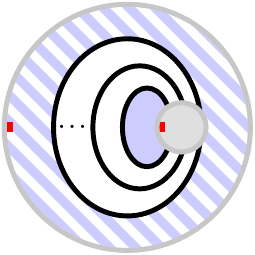}}, &&
        \tr_{n,+}^{(r)}\!:=\!\raisebox{-1.18cm}{\includegraphics{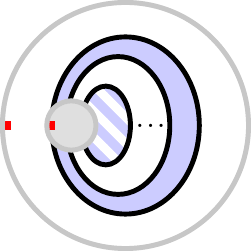}}, &&
        \tr_{n,-}^{(r)}\!:=\!\raisebox{-1.18cm}{\includegraphics{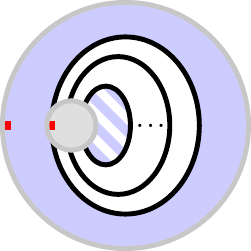}},
\end{align}
induce notions of \textit{left} and \textit{right traces}, respectively:
\begin{align}
 \Pb_{\tr_{n,\pm}^{(l)}}: A_{n,\pm} \to A_{0,\pm(-)^n}, \qquad
 \Pb_{\tr_{n,\pm}^{(r)}}: A_{n,\pm} \to A_{0,\pm}.
\end{align}
A planar algebra is said to be \textit{spherical} if
\be\label{equ:DiagTrace}
 \Pb_{\tr_{n,\pm}^{(l)}}=\Pb_{\tr_{n,\pm}^{(r)}}
\ee
for all $n\in\Nbb$. We note that sphericality requires $A_{0,+}\cong A_{0,-}$.

Let $\cdot^\dagger$ denote the operator that acts by reflecting a planar tangle about a line perpendicular to 
the marked exterior boundary interval, and let $^*:A_{n,\pm}\to A_{n,\pm}$, $n\in\Nbb_0$, denote a conjugate 
linear involution. Analogous to naturality, compatibility between the two maps manifests itself in a simple relation,
\begin{align}\label{equ:InvNat}
    \Pb_{T^{\dagger}}(v_1^*,\ldots,v_{|\mathcal{D}_T|}^*) 
    = \Pb_{T}(v_1,\ldots,v_{|\mathcal{D}_T|})^*,
\end{align}
which must hold for all planar tangles $T$ and all 
$(v_1,\ldots,v_{|\mathcal{D}_T|})\in\bigtimes_{D\in\mathcal{D}_{T}}A_{\eta(D)/2,\,\zeta(D)}$. 
A planar algebra $(A_{n,\pm})_{n\in\Nbb_0}$ endowed with the maps $\cdot^{\dagger}$ and $\cdot^{*}$
satisfying (\ref{equ:InvNat}) is known as \textit{involutive}. In that case,
\begin{align}\label{equ:IdnConj}
    \Pb_{\mathrm{Id}_{n,\pm}^{\dagger}}() = \Pb_{\mathrm{Id}_{n,\pm}}()^*,\qquad 
    \mathds{1}_{n,\pm} =  \mathds{1}_{n,\pm}^*,
\end{align}
and for $p\in A_{n,\pm}$, we have
\be
 p^2=p\qquad\Longrightarrow\qquad (p^*)^2=p^*,
\label{pstar}
\ee
with the indicated multiplication induced by $M_{n,\pm}$.

An involutive planar algebra $(A_{n,\pm})_{n\in\Nbb_0}$ admits the sesquilinear maps
\begin{align}\label{Markov}
    \langle\cdot\,, \cdot\rangle_{n,\pm}^{(c)}: A_{n,\pm}\times A_{n,\pm} \to A_{0,\pm(-)^{n\delta_{\scalebox{0.8}{\tiny $c,l$}}}}, 
    \qquad 
    (a,b)\mapsto \Pb_{\mathrm{tr}_{n,\pm}^{(c)}}(a^*b),
\end{align}
labelled by $c\in\{l,r\}$ and known as the \textit{left} and \textit{right trace map} for $c=l$ and $c=r$, 
respectively. Here, $\delta_{c,\, l}=1$ for $c=l$, and $\delta_{c,\, l}=0$ for $c=r$.
Composed with the evaluation map (\ref{eval}) for $\dim A_{0,\pm}=1$, we obtain the sesquilinear
\textit{trace forms}
\be
 \mathrm{e}\circ\langle\cdot\,, \cdot\rangle_{n,\pm}^{(c)}:A_{n,\pm}\times A_{n,\pm} \to\Cbb.
\label{form}
\ee
An involutive planar algebra inherits the qualifier \textit{positive (semi-)definite} 
if both trace forms enjoy it for all $n\in\Nbb_0$. If the involutive planar algebra is spherical, then the two
trace maps (and hence trace forms) are identical.

An evaluable, spherical and positive-definite planar algebra is called a \textit{subfactor planar algebra}.
It is implied that a positive-definite planar algebra is involutive, so a subfactor planar algebra 
$(A_{n,\pm})_{n\in\Nbb_0}$ is involutive, with each $A_{n,\pm}$ having the corresponding unique trace form 
\eqref{form} as inner product. In that case, we introduce
\be\label{equ:TrDash}
    \Pb'_{\mathrm{tr}_{n,\pm}^{(c)}}:=\mathrm{e}\circ\Pb_{\mathrm{tr}_{n,\pm}^{(c)}}
\ee
and refer to the map
\begin{align}\label{norm}
    A_{n,\pm} \to \Rbb, 
    \qquad 
    a\mapsto\sqrt{\Pb'_{\mathrm{tr}_{n,\pm}^{(c)}}(a^*a)},
\end{align}
as the \textit{trace norm}.
Moreover, the multiplication on $A_{n,\pm}$ is induced by the corresponding planar tangle
in \eqref{equ:MultAndCoMultPlanar}, and since evaluability implies $\dim A_{n,\pm}<\infty$, 
it follows that $A_{n,\pm}$ is a finite-dimensional semisimple algebra, see e.g.~\cite{Jones00}. 
Consequently, $A_{n,\pm}$ is isomorphic to a direct sum of matrix algebras, facilitating the use of 
linear-algebraic techniques in the analysis of subfactor planar algebras. As a simple but ubiquitous example of a 
subfactor planar algebra, we recall the \textit{Temperley--Lieb subfactor planar algebra} \cite{Jones99}
in Section \ref{subsec:TLPA} below.
\medskip

\noindent 
\textbf{Remark.} 
The positive-definiteness of subfactor planar algebras implies that they do not have any null vectors. 
Given the relevance of subfactor planar algebras to the classification in Theorem \ref{thm:1}, we will henceforth 
assume that $\Pb_{\mathrm{id}_{n,\pm}}$ is the identity operator and that the algebra with multiplication 
induced by $M_{n,\pm}$ is unital.

\subsection{Yang--Baxter relation planar algebras}
\label{Sec:YBRpa}

Here, we recall a variant of subfactor planar algebras that will play an important role in our discussion of 
Yang--Baxter integrability. 
Let $(A_{n,\pm})_{n\in\Nbb_0}$ be a subfactor planar algebra, with $B_{2,\pm}$ denoting a basis for $A_{2,\pm}$. 
Following \cite{Liu15}, a triple $(x,y,z)\in A_{2,-}\times A_{2,+} \times A_{2,-}$, respectively 
$(x,y,z)\in A_{2,+}\times A_{2,-} \times A_{2,+}$, is said to satisfy a \textit{Yang--Baxter relation} (YBR) if
\begin{align}\label{ybrdec}
 \raisebox{-1cm}{\ShadedPosYBRLHS} =\sum_{\substack{a,c\in B_{2,+}\\ b\in B_{2,-}}}
  \!\!C_{x,y,z}^{a,b,c}\
  \raisebox{-1cm}{\ShadedPosYBRRHS}\;, 
  &&\
 \raisebox{-1cm}{\ShadedNegYBRLHS} =\sum_{\substack{a,c\in B_{2,-}\\ b\in B_{2,+}}}
  \!\!D_{x,y,z}^{a,b,c}\
  \raisebox{-1cm}{\ShadedNegYBRRHS}\;,
\end{align}
respectively, for some $C_{x,y,z}^{a,b,c},D_{x,y,z}^{a,b,c}\in\Cbb$. A subfactor planar algebra is called a 
\textit{YBR planar algebra} if every triple of vectors in $A_{2,-}\times A_{2,+} \times A_{2,-}$ 
and $A_{2,+}\times A_{2,-} \times A_{2,+}$ satisfies a YBR.
\medskip

\noindent
\textbf{Remark.}
Although a YBR planar algebra is a \textit{subfactor} planar algebra, we are suppressing that qualifier,
in line with the convention in \cite{Liu15}. 
\medskip

While we adopt the form (\ref{ybrdec}) of the YBRs introduced in \cite{Liu15}, 
the invertibility of the elementary rotation tangles (see Corollary \ref{cor:ElRot} in Appendix \ref{app:pa})
ensures that the characterisation of a planar algebra as a YBR planar algebra does not depend on the particular 
choices of input disk markings (and consequently shadings) in (\ref{ybrdec}), on either side of any of the two YBRs. 
However, we do need a YBR for each shading of the \textit{output} disk, as in (\ref{ybrdec}).

In Section \ref{Sec:TYB}, we describe how one can associate a 
\textit{homogeneous Yang--Baxter integrable model} to any planar algebra satisfying a particular set of sufficient 
conditions, including YBRs. It is thus natural to expect that YBR planar algebras play an important role in the 
classification of Yang--Baxter integrable planar-algebraic models. 
Indeed, we find (Proposition \ref{prop:SGYBIntYBR} in Section \ref{Sec:Baxt}) that a so-called 
\textit{singly generated planar algebra} that is \textit{not} a YBR planar algebra \textit{does not} encode the 
structure of a homogeneous Yang--Baxter integrable model. 
Singly generated planar algebras are discussed in Section \ref{Sec:YBIntYBR}.

\subsection{Temperley--Lieb planar algebra}
\label{subsec:TLPA}

Let $\mathrm{T}_{n,\pm}$ denote the complex vector space spanned by disks with 
(i) $2n$ nodes where each node is connected to another node via a set of non-intersecting strings 
(defined up to ambient isotopy), and 
(ii) a $\pm$ checker-board shading (see Section \ref{subsec:YBRPA}). 
Examples of such \textit{Temperley--Lieb disks} are
\begin{align}
    \raisebox{-0.45cm}{\TLcircExIShaded} \in \mathrm{T}_{2,-},\qquad\quad
    \raisebox{-0.45cm}{\TLcircExIIShaded} \in \mathrm{T}_{3,+},\qquad\quad
    \raisebox{-0.45cm}{\TLcircExIIIShaded} \in \mathrm{T}_{5,-}.
\end{align}
The dimension of $\mathrm{T}_{n, \pm}$ is given by a Catalan number, as
\begin{align}
 \dim \mathrm{T}_{n,\pm}=\frac{1}{n+1}\binom{2n}{n}.
\end{align}

The \textit{Temperley--Lieb} (TL) \textit{planar algebra} $\mathrm{TL}(\delta)$ is the graded vector space 
$(\mathrm{T}_{n,\pm})_{n\in\mathbb{N}_0}$, together with the natural diagrammatic action of shaded planar 
tangles, with each loop replaced by a factor of the parameter $\delta\in\mathbb{C}$. 
To illustrate, we present the example
\begin{align}
    T =  \raisebox{-1.4275cm}{\FCExTangleAltShaded}\,,\qquad\quad
    \Pb_T\big(\!\raisebox{-0.175cm}{\RotRopFCxESTBlackAltColouredUnshadedShaded}\,,\!
    \raisebox{-0.175cm}{\RotRopFCxXSTBlackAltColouredUnshadedShaded}\,\big) 
    \,= \raisebox{-1.4275cm}{\FCExTangleInputAltColouredUnshadedShaded} 
    \,=\, \delta\raisebox{-0.43cm}{\FCExTangleInputIdColouredUnshadedShaded}\;.
\end{align}
From \cite{Jones99}, we know that $\mathrm{TL}(\delta)$ is spherical and involutive, 
with the involution $\cdot^*$ defined as the conjugate-linear map that acts by reflecting every disk about a line 
perpendicular to the marked boundary interval.

Let $\mathcal{T}$ denote the set of all $\delta$-values such that the planar algebra $\mathrm{TL}(\delta)$ is 
positive semi-definite. For each $\delta\in\mathcal{T}$, the \textit{TL subfactor planar algebra} 
$\mathsf{TL}(\delta)$ is then defined as the quotient of $\mathrm{TL}(\delta)$ by the kernel of the trace 
norm \eqref{norm}. According to \cite{Jones83},
\begin{align}\label{TTL}
    \mathcal{T} = \big\{2\cos\tfrac{\pi}{m+2}\,|\, m\in\mathbb{N}\big\}\cup [2,\infty).
\end{align}
For $\delta\in\big\{2\cos\tfrac{\pi}{m+2}\,|\, m\in\mathbb{N}\big\}$, the kernel of the trace norm is nontrivial, so 
$\mathrm{TL}(\delta)$ and $\mathsf{TL}(\delta)$ are non-isomorphic, while for $\delta>2$, 
$\mathrm{TL}(\delta)$ is positive-definite, in which case $\mathrm{TL}(\delta) \cong \mathsf{TL}(\delta)$. 

\medskip

\noindent
\textbf{Remark.} Here and throughout, the \textsf{sans-serif} font distinguishes a subfactor planar algebra,
such as $\mathsf{TL}(\delta)$, from the corresponding (not necessarily subfactor) planar algebra, here $\mathrm{TL}(\delta)$.
\medskip

Every subfactor planar algebra has a planar subalgebra isomorphic to the
Temperley--Lieb subfactor planar algebra \cite{Jones99, Jones83}. Accordingly, $A_{2,\pm}$ of a subfactor 
planar algebra  $(A_{n,\pm})_{n\in\Nbb_0}$ contains the two Temperley--Lieb disks
\begin{align}\label{TLvectors}
    \mathds{1}_{2,\pm}=\Pb_{\mathrm{Id}_{2,\pm}}(),\qquad
    e_{\pm}:=\Pb_{\Ec_{\pm}}(),
\end{align}
where
\begin{align}
 \mathrm{Id}_{2,+}=\raisebox{-0.285cm}{\TLAidIITangShadePos},\qquad
 \mathrm{Id}_{2,-}=\raisebox{-0.285cm}{\TLAidIITangShadeNeg},\qquad
 \Ec_{+}:=\raisebox{-0.285cm}{\TLAeProjTangShadeNeg},\qquad
 \Ec_{-}:=\raisebox{-0.285cm}{\TLAeProjTangShadePos},
\end{align}
with $\mathrm{Id}_{2,\pm}$ a special case of (\ref{Aunit}). In all but the degenerate case (see Remark after
(\ref{Un})), which we exclude in the following, the vectors in \eqref{TLvectors} are linearly independent so 
$\dim A_{2,\pm}\geq2$, while the Temperley--Lieb subfactor planar algebra itself has $\dim A_{2,\pm}=2$.

In a so-called \textit{singly generated} (subfactor) \textit{planar algebra}, $\dim A_{1,\pm}=1$ and 
$\dim A_{2,\pm}=3$, so $A_{2,\pm}$ has a basis consisting of the two Temperley--Lieb disks (\ref{TLvectors}) 
and \textit{one} additional disk, hence the terminology. Moreover, the vector spaces $A_{n,\pm}$ for $n>2$ are 
generated by the action of planar tangles on disks in $A_{2,\pm}$. Section \ref{Sec:YBIntYBR} below is devoted to 
the study of (unshaded) singly generated planar algebras.
\medskip

\noindent
\textbf{Remark.}
Although a singly generated planar algebra is a \textit{subfactor} planar algebra, we are suppressing that qualifier, 
in line with the convention in \cite{BJ00}.

\section{Proto-singly-generated algebra}
\label{Sec:YBIntYBR}

Subfactor planar algebras, including singly generated planar algebras, are shaded by default. 
The shading is necessary, in general, to relate a subfactor planar algebra 
to the so-called \textit{standard invariant} of the corresponding \textit{subfactor}; this story is outlined in \cite{Bisch02}.
However, the shading of a planar algebra $(A_{n,\pm})_{n\in\mathbb{N}_0}$ need not carry any nontrivial 
information. In that case, the shading can be ignored, giving rise to the corresponding
\textit{unshaded planar algebra} $(A_n)_{n\in\mathbb{N}_0}$, see e.g.~\cite{LMP17}.
\medskip

\noindent
\textbf{Remark.}
Omitting the subscript indicating the shading of a given shaded planar tangle, we are referring
to the corresponding unshaded version of the planar tangle. A similar convention is adopted for
the vector spaces $(A_n)_{n\in\mathbb{N}_0}$ spanned by unshaded disks.
\medskip

As will be clear after Section \ref{Sec:Unshaded}, the singly generated planar algebras relevant to the 
integrability questions of our interest necessarily admit an unshaded description. 
Accordingly, we will henceforth restrict to the class of \textit{unshaded} singly generated planar algebras.
About these, we have the following key result involving the \textit{proto-singly-generated planar algebra} 
$\mathrm{PS}^{(\epsilon)}(\alpha,\delta)$ constructed in Section \ref{Sec:Proto} and Section \ref{Sec:Gen} below.
This algebra was conceived in \cite{BJ00}.
\begin{prop}\label{prop:UnSGeqPSG}
    An unshaded singly generated planar algebra is a quotient of $\mathrm{PS}^{(\epsilon)}(\alpha,\delta)$
    for some $\epsilon,\alpha,\delta$.
\end{prop}
\noindent
Indeed, $\mathrm{PS}^{(\epsilon)}(\alpha,\delta)$ is defined in terms of relations that, 
for some $\epsilon,\alpha,\delta$, are satisfied by any given unshaded singly generated planar algebra.
Here, $\delta$ is the loop weight arising in
\be
 e^2=\delta e,
\ee
where
\be
 e:=\Pb_{\Ec}(),\qquad \Ec:=\raisebox{-0.285cm}{\TLAeProjTang},
\ee
is the unshaded Temperley--Lieb generator. We note that $e^*=e$.

\subsection{Planar algebra}
\label{Sec:Proto}

We now introduce a planar algebra whose vector spaces are spanned by planar tangles with labelled disks, 
and where planar tangles act on these vector spaces in the natural way.
For each $n\in\Nbb_0$, let $S_n$ denote a set whose elements label disks with 
$2n$ nodes. With $S := \bigsqcup_{n\in\Nbb_0}S_n$, an \textit{$S$-labelled tangle} is thus a planar tangle 
whose input disks each have been labelled by an element of $S$. 
For such a label set $S$, the \textit{unshaded universal planar algebra} 
consists of the vector spaces $(A_n(S))_{n\in \Nbb_0}$ where, for each $n$, $A_n(S)$ is spanned by all 
$S$-labelled tangles with $2n$ nodes on their output disk, together with the planar-tangle action 
colloquially named `what you see is what you get', illustrated in \eqref{equ:SGTangActionEx}, 
see also \cite{Jones99,KT11}. We note that the elements of $S$ have no additional structure. Accordingly, 
the list of cardinalities, $|S_0|,|S_1|,|S_2|,\ldots$, characterises a universal planar algebra.

A universal planar algebra is infinite-dimensional. Indeed, even if $S_k=\emptyset$ for all $k\in\Nbb_0$, 
then each $A_n(S)$ is spanned by the corresponding set of (unshaded) Temperley--Lieb disks, 
together with those same disks but with all possible combinations of loops. To tame the dimensionality of 
a universal planar algebra, relations are imposed on $(A_n(S))_{n\in \Nbb_0}$. For any set $C$ of vectors in 
$(A_n(S))_{n\in \Nbb_0}$, we thus let $I(C)$ denote the planar ideal generated by $C$, 
and $(A_n(S,C))_{n\in \Nbb_0}$ the corresponding \textit{quotient planar algebra}.

We refer to a (universal) quotient planar algebra $(A_n(S,C))_{n\in \Nbb_0}$ as a 
\textit{proto-singly-generated} (PSG) \textit{planar algebra} if $S$ and $C$ satisfy
\begin{align}\label{AAA}
    \mathrm{dim}\,A_0(S,C) = \mathrm{dim}\,A_1(S,C) = 1, \qquad
    \mathrm{dim}\,A_2(S,C) = 3,\qquad
    |S_n| = \delta_{n,2},
\end{align}
such that $A_0(S,C)$, $A_1(S,C)$, and $A_2(S,C)$ are as in an unshaded singly generated planar algebra. 
Accordingly, these vector spaces are (i) `spherical': satisfying \eqref{equ:DiagTrace} for $n = 0,1,2$; 
(ii) `involutive': closed under $\cdot^*$, satisfying (\ref{equ:InvNat}) for all planar tangles $T$ with 
$\eta(D)\in\{0,2,4\}$ for all $D\in\mathcal{D}_T\cup \{D_0^T\}$; and
(iii) `positive-definite': the trace form \eqref{form} being positive-definite for $n = 0,1,2$.
The PSG planar algebra is otherwise generated by the action of the planar tangles, 
with no further relations imposed.
\medskip 

\noindent 
\textbf{Remark.} 
A PSG planar algebra is not evaluable (since $\dim A_n(S,C)=\infty$ for $n>2$) nor necessarily having
a positive-definite trace form for each $n$. It follows that additional structure must be imposed on $A_n(S,C)$ for 
$n>2$ to obtain a singly generated planar algebra. 
\medskip

Since $A_1(S,C)$ is positive-definite, we have
\be
 0<\Pb'_{\tr_1}(\mathds{1}_1^*\mathds{1}_1) 
 =\delta,
\label{tr1}
\ee
where $\Pb'_{\tr_1}$ is defined as in \eqref{equ:TrDash}. 
With that, the positive-definiteness of $A_2(S,C)$ similarly implies that
\be
 0<\Pb'_{\tr_2}\big([\mathds{1}_2-\tfrac{1}{\delta}e]^*[\mathds{1}_2-\tfrac{1}{\delta}e]\big)
  =\delta^2-1,
\label{tr2}
\ee
from which it then follows that
\be
 \delta>1.
\label{d1}
\ee
The evaluations in (\ref{tr1}) and (\ref{tr2}) involve the Jones--Wenzl idempotent $\wj_n$ (\ref{wj})
for $n=1$ and $n=2$, respectively: $\wj_1 = \mathds{1}_1$ and $\wj_2 = \mathds{1}_2 - \frac{1}{\delta}e$.
For general $n\in\Nbb$, we have $\wj_n^*=\wj_n$ and
\be
  \Pb'_{\tr_n}(\wj_n^*\wj_n)
  =\Pb'_{\tr_n}(\wj_n)
  =U_n\big(\tfrac{\delta}{2}\big)
   =\prod_{k=1}^n\Big(\delta-2\cos\frac{k\pi}{n+1}\Big),
\label{Un}
\ee
where $U_n$ is the $n^{\mathrm{th}}$ Chebyshev polynomial of the second kind.
We note that the first equality in (\ref{Un}) is a special case of Lemma \ref{Lem:ppp} below.
\medskip

\noindent
\textbf{Remark.}
If $(A_n(S,C))_{n\in\mathbb{N}_0}$ is positive semi-definite for a given value of $\delta$, 
then one obtains a well-defined subfactor planar algebra by quotienting out the ideal generated by all the 
vectors $a\in A_n(S,C)$ for which $\Pb'_{\tr_n}(a^*a)=0$ for all $n\in\mathbb{N}_0$. 
In the degenerate case $\delta=1$, for example, the Temperley--Lieb planar subalgebra is 
trivialised by quotienting out the ideal generated by $\mathds{1}_2-\tfrac{1}{\delta}e$
appearing in (\ref{tr2}), whereby $\mathrm{dim}\,A_2(S,C)$ reduces to $2$, c.f.~(\ref{AAA}).
\medskip

The conditions $C$ are determined in Section \ref{Sec:Gen} below, where we find that the PSG planar algebra
is unique, up to specification of parameters, including the loop weight $\delta$.
From here onward, we accordingly opt for the abridged notation $\mathrm{PS}_n\equiv A_n(S,C)$, $n\in\Nbb_0$,
with $S$ and $C$ as above.

\subsection{Defining relations}
\label{Sec:Gen}

Here, we determine the relations satisfied by the vectors in a distinguished $\mathrm{PS}_2$-basis of the form 
$\{\mathds{1}_2, e, s\}$. Taking inspiration from the classification approach in \cite{Liu15}, selecting $s$ as 
conveniently as possible is key in the following.
\begin{lem}\label{Lem:ppp}
Let $(A_{n,\pm})_{n\in\Nbb_0}$ be a subfactor planar algebra and $\{p_1,\ldots,p_k\}$ a basis for $A_{n,\pm}$ 
for some $n$ and shading $+/-$, and suppose $\{p_1,\ldots,p_k\}$ is a complete set of mutually orthogonal 
idempotents, with multiplication induced by $M_{n,\pm}$. 
Then, $\Pb'_{\tr_{n,\pm}}(p_i)>0$ and $p_i^*=p_i$ for all $i=1,\ldots,k$.
\end{lem}
\begin{proof}
By construction,
\be
 p_i^*=\sum_{j=1}^kc_{ij}p_j,\qquad i=1,\ldots,k,
\ee
for some scalars $c_{ij}$, where, by (\ref{pstar}),
\be
 c_{ij}\in\{0,1\},\qquad \forall\,i,j\in\{1,\ldots,k\},
\ee
while the positive-definiteness implies that
\be
 0<\Pb'_{\tr_{n,\pm}}(p_i^*p_i)=c_{ii}\Pb'_{\tr_{n,\pm}}(p_i),\qquad i=1,\ldots,k,
\ee
so $c_{11}=\cdots=c_{kk}=1$ and
\be
 \Pb'_{\tr_{n,\pm}}(p_i)>0,\qquad i=1,\ldots,k.
\ee
From
\be
 0=\mathds{1}_{n,\pm}^*-\mathds{1}_{n,\pm}=\sum_{i,j=1}^kc_{ij}p_j-\sum_{j=1}^kp_j,
\ee
it then follows that $c_{ij}=\delta_{ij}$, hence $p_i^*=p_i$ for all $i=1,\ldots,k$.
\end{proof}

Now, let $\mathrm{PS}_2$ be endowed with the multiplication induced by the \textit{unshaded} companion 
to \eqref{equ:MultAndCoMultPlanar}. Since $\dim \mathrm{PS}_0=1$, the idempotent 
$\ip_0:=\frac{1}{\delta}e$ satisfies $\dim(\ip_0\mathrm{PS}_2\ip_0)=1$ and is therefore primitive. 
By assumption, $\mathrm{PS}_2$ is positive-definite, hence semisimple, and because 
$\{\mathds{1}_2,\ip_0\}\subset \mathrm{PS}_2$ and $\dim\mathrm{PS}_1=1$, it follows that 
$\mathrm{PS}_2$ is commutative, see e.g.~\cite{JonesNotes}. The set $\{\ip_0\}\subset \mathrm{PS}_2$ can thus 
be extended to a $\mathrm{PS}_2$-basis, $\{\ip_0,\ip_1,\ip_2\}$, consisting of a complete set of mutually 
orthogonal (and primitive) idempotents, where we note that
\be
 \delta^2
 =\Pb'_{\tr_2}(\mathds{1}_2)
 =1+\Pb'_{\tr_2}(\ip_1)+\Pb'_{\tr_2}(\ip_2).
\ee
By Lemma \ref{Lem:ppp}, we also have
\be
 \Pb'_{\tr_2}(\ip_1)>0,\qquad \Pb'_{\tr_2}(\ip_2)>0.
\label{PsP}
\ee

We now introduce
\begin{align}
\label{SPQ}
 s= p_1\ip_1 + p_2\ip_2,
\end{align}
where $p_1,p_2\in\mathbb{C}$, and for $\{\mathds{1}_2,e,s\}$ to be a basis for $\mathrm{PS}_2$, it must hold 
that $p_1\neq p_2$. It follows that
\begin{align}\label{equ:QuadRelsSG}
 es=0=se,\qquad s^2 = (p_1+p_2)s -p_1p_2\big(\mathds{1}_2 -\tfrac{1}{\delta}e\big),
\end{align}
hence $\Pb^2_{r_{2,1}}(s)e=0$, and that
\be
 \Pb_{\tr_2}(s)=p_1\Pb_{\tr_2}(\ip_1)+p_2\Pb_{\tr_2}(\ip_2).
\ee
For convenience, we choose $p_1,p_2$ such that
\be
 \Pb_{\tr_2}(s)=0,
\label{trs}
\ee
noting that (\ref{PsP}) then ensures that $p_1\neq p_2$ (as required) and implies that $p_1,p_2\neq0$. Without 
loss of generality, we may further choose the normalisation of $s$ such that $p_1p_2=-1$, thereby obtaining
\be
 s^2=\mathds{1}_2-\tfrac{1}{\delta}e+\alpha s,
\ee
where $\alpha:=p_1+p_2$, noting that $\alpha$ can take any value in 
$\Cbb\setminus\{-2\mathrm{i},2\mathrm{i}\}$.

By construction,
\be
 \Pb_{r_{2,1}}(s)=\epsilon_{\mathds{1}}\mathds{1}_2+ \epsilon_e e + \epsilon s
\ee
for some $\epsilon_{\mathds{1}},\epsilon_e,\epsilon\in\Cbb$, while (\ref{trs}) implies that
\be
 \Pb_{r_{2,1}}(s)e=0
\label{se0}
\ee
(since $\delta\neq0$) and, by sphericality, that $\Pb^3_{r_{2,1}}(s)e=0$.
Using $\Pb_{r_{2,1}}^4\!=\mathrm{id}$ and $\delta^2\neq1$, it follows that
\be
 \epsilon_{\mathds{1}}=\epsilon_e=0,\qquad \epsilon^4=1.
\label{ccc}
\ee
Subsequently, recalling that $\Pb_{r_{2,2}}=\Pb_{r_{2,1}}^2$, the relations
$\Pb_{r_{2,2}}(\mathds{1}_2)=\mathds{1}_2$ and $\Pb_{r_{2,2}}(e)=e$ imply that
\be
 0=\Pb_{r_{2,2}}(s^2)-\Pb_{r_{2,2}}(s)\,\Pb_{r_{2,2}}(s)=\alpha(\epsilon^2-1)s,
\ee
so we must have $\alpha=0$ if $\epsilon^2=-1$. This requirement may be implemented by setting
\be
 \alpha=(1+\epsilon^2)\hat{\alpha},
\ee
where $\hat{\alpha}\in\Cbb\setminus\{-\mathrm{i},\mathrm{i}\}$. 

Under the conjugate linear involution $\cdot^*$, we have
\begin{align}\label{conj}
    \mathds{1}_2^*=\mathds{1}_2,\qquad 
    e^* = e,\qquad 
    s^* =\frac{p_1\bar{p}_2-\bar{p}_1p_2}{p_1-p_2}\big(\mathds{1}_2 - \tfrac{1}{\delta}e\big) 
    +\frac{\bar{p}_1-\bar{p}_2}{p_1-p_2}\,s,
\end{align}
where $\bar{p}$ denotes the complex conjugate to the scalar $p$.
Using $p_1p_2 = -1$, it follows that, for $\alpha\in\mathbb{R}$, we have $p_1,p_2\in\mathbb{R}$, hence $s^*=s$.

Expressing $s$ diagrammatically as
\be
 s=\!\raisebox{-0.285cm}{\SOpSG} ,
\ee
the analysis above implies that
\begin{align}\label{equ:SQuadSG}
    \raisebox{-0.325cm}{\LLbraidAlt}
    = \raisebox{-0.285cm}{\TLAidII}  
    \!- \tfrac{1}{\delta}\raisebox{-0.285cm}{\TLAeProj}
    \!+\alpha\raisebox{-0.285cm}{\SOpSG} ,\qquad
    \raisebox{-0.285cm}{\SOpSGRotI} \!= \epsilon\raisebox{-0.285cm}{\SOpSG}, \qquad 
    \epsilon \in \{-1,1,-\mathrm{i},\mathrm{i}\}.
\end{align}
For each $(\alpha,\epsilon)$ and $\delta>1$, where 
$(\alpha,\epsilon)\in\big(\Cbb\setminus\{-2\mathrm{i},2\mathrm{i}\}\big)\times\{-1,1\}$ 
or $(\alpha,\epsilon)\in\{0\}\times\{-\mathrm{i},\mathrm{i}\}$, we thus have a PSG planar algebra 
$(\mathrm{PS}_n)_{n\in\Nbb_0}$, where a basis for $\mathrm{PS}_2$ is given diagrammatically by
\begin{align}\label{equ:DiagBasisSG}
   B_2 = \left\{\!\raisebox{-0.285cm}{\TLAidII},
   \raisebox{-0.285cm}{\TLAeProj},
   \raisebox{-0.285cm}{\SOpSG} \!\right\} .
\end{align}
To illustrate the corresponding action of planar tangles, we have
\begin{align}\label{equ:SGTangActionEx}
    T =  \raisebox{-1.4275cm}{\PGExTangleAlt}\,,\qquad
    \Pb_T\big(\!\raisebox{-0.179cm}{\RotRopTLxBRSTAltAlt}, 
    \raisebox{-0.179cm}{\RotRopTLxESTAltAltAlt}, \raisebox{-0.179cm}{\RotRopTLxESTAltAltVIII}\,\big) 
    = \raisebox{-1.4275cm}{\PGExTangleInputAlt} \,= \delta \raisebox{-0.43cm}{\PGExTangleInputId}\;.
\end{align}

We note that the conditions 
\be
 \Pb^\ell_{r_{2,1}}(s)e=0,\qquad \ell=0,1,2,3,
\label{Pse}
\ee
where $\Pb^0_{r_{2,1}}$ is the identity map on $\mathrm{PS}_2$, correspond to the 
following diagrammatic relations in $\mathrm{PS}_1$:
\vspace{-.1cm}
\be
 \raisebox{-0.43cm}{\SOpSGCappedIAlt} \!=\!
 \raisebox{-0.43cm}{\SOpSGCappedIIIAlt}  \!=\! 
 \raisebox{-0.43cm}{\SOpSGCappedIIAlt}  \!=\! 
 \raisebox{-0.43cm}{\SOpSGCappedIVAlt}  \!=\, 0.
\ee
In fact, this property of $s$ was a motivating factor behind (\ref{trs}). Moreover, the tracelessness of $s$ 
allows us to represent the trace-form inner product relative to the ordered basis $\{\mathds{1}_2,e,s\}$ as
\be
 \begin{pmatrix} \delta^2&\delta&0\\ \delta&\delta^2&0\\ 0&0&\delta^2-1\end{pmatrix},
\ee
confirming the positive-definiteness of $\mathrm{PS}_2$ for $\delta>1$, c.f.~(\ref{d1}).

To summarise, the PSG planar algebra $\mathrm{PS}^{(\epsilon)}(\alpha,\delta)$
is the quotient planar algebra $(A_n(S,C^{(\epsilon)}(\alpha, \delta)))_{n\in\mathbb{N}_0}$, where 
\be
 S= \bigsqcup_{n\in\Nbb_0}S_n,\qquad 
 S_2 = \{\!\!\!\raisebox{-0.1cm}{\scalebox{1.75}{\PSGLabelII}}\!\!\!\},\qquad
 S_k = \emptyset,\qquad 
 \forall\,k\neq 2, 
\label{equ:PSGLabs}
\ee
and
\be
    C^{(\epsilon)}(\alpha, \delta)= \Big\{\raisebox{-0.285cm}{\BMWloopAlt} 
    \!-\delta\raisebox{-0.285cm}{\BMWemptyAlt},\!\raisebox{-0.285cm}{\LtwistIIAltUni}, 
    \raisebox{-0.285cm}{\SOpSGRotI} \!- \epsilon\raisebox{-0.285cm}{\SOpSG}, 
    \raisebox{-0.325cm}{\LLbraidAlt}
    \!-\raisebox{-0.285cm}{\TLAidII}  
    \!+\tfrac{1}{\delta}\raisebox{-0.285cm}{\TLAeProj}
    \!-\alpha\raisebox{-0.285cm}{\SOpSG} \!\Big\},
\label{equ:PSGRels}
\ee
with $\delta>1$, and $(\alpha,\epsilon)\in\big(\Cbb\setminus\{-2\mathrm{i},2\mathrm{i}\}\big)\times\{-1,1\}$ or 
$(\alpha,\epsilon)\in\{0\}\times\{-\mathrm{i},\mathrm{i}\}$. 

With the parameters as above, $\mathrm{PS}^{(\epsilon)}(\alpha,\delta)$ is the unique
planar algebra satisfying the conditions outlined in the paragraph containing \eqref{AAA}.
It follows that any unshaded singly generated planar algebra can be obtained by possibly specialising the 
parameters $\epsilon,\alpha,\delta$ and by taking a quotient in such a way that each $\mathrm{PS}_n$, 
$n\in \{3,4,\ldots\}$, is spherical, involutive and positive-definite. 
These observations conclude the proof of Proposition \ref{prop:UnSGeqPSG}.

\subsection{Presentation}
\label{Sec:PSpresentation}

We proceed to describe the algebra that arises when endowing the vector space $\mathrm{PS}_n$ with 
the multiplication induced by the unshaded companion to \eqref{equ:MultAndCoMultPlanar}.
For each $n\in\Nbb$, $\delta > 1$, and each pair 
$(\alpha,\epsilon)\in\big(\Cbb\setminus\{-2\mathrm{i},2\mathrm{i}\}\big)\times\{-1,1\}$
or $(\alpha,\epsilon)\in\{0\}\times\{-\mathrm{i},\mathrm{i}\}$, the \textit{proto-singly-generated} (PSG) 
\textit{algebra} $\mathrm{PS}_n^{(\epsilon)}(\alpha,\delta)$ is thus defined as the unital 
associative algebra $\langle e_i,s_i\,|\,i=1,\ldots,n-1\rangle$ subject to the relations
\be
\begin{array}{rllll}
  & s_i^2=\mathds{1}-\tfrac{1}{\delta}e_i+\alpha s_i,\qquad 
  & e_is_i= s_ie_i = 0,\qquad
  & s_is_j=s_js_i,\qquad
  &  |i-j|>1,
\\[.15cm]
  & e_ie_{i\pm1}e_i= e_i,\qquad
  & e_ie_{i\pm1}s_i=\epsilon^{\pm1}e_is_{i\pm1},\qquad
  & s_ie_{i\pm1}e_i=\epsilon^{\mp1}s_{i\pm1}e_i,
  &
\end{array}
\label{PSrel}
\ee
with $\mathds{1}$ denoting the unit. Following from (\ref{PSrel}), we also have the relations
\be
 e_i^2=\delta e_i,\qquad 
 e_ie_j=e_je_i,\qquad 
 e_is_j=s_je_i,\qquad
 |i-j|>1,
\ee
\be
  e_is_{i\pm1}s_i=e_i\big(\epsilon^{\mp1}(e_{i\pm1}-\tfrac{1}{\delta}\mathds{1})+\alpha s_{i\pm1}\big),\qquad
  s_is_{i\pm1}e_i=\big(\epsilon^{\pm1}(e_{i\pm1}-\tfrac{1}{\delta}\mathds{1})+\alpha s_{i\pm1}\big)e_i,
\ee
and
\be
  e_is_{i\pm1}e_i=0,\qquad s_ie_{i+1}s_i=s_{i+1}e_is_{i+1}.
\ee
We note that it suffices to list one of the two relations $e_is_i= 0$ and $s_ie_i = 0$ in (\ref{PSrel}).

The generators of $\mathrm{PS}_n^{(\epsilon)}(\alpha,\delta)$ are represented diagrammatically as
\begin{align}
    \mathds{1}\,\leftrightarrow \raisebox{-0.885cm}{\BMWRecIdOdd},\qquad
    e_i\,\leftrightarrow \raisebox{-0.885cm}{\BMWRecEOdd},\qquad 
    s_i\,\leftrightarrow \raisebox{-0.885cm}{\LRecBOdd},
\label{1es}
\end{align}
with the marked boundary interval linking the two horizontal edges via an invisible arc on the left.
Multiplication is then implemented by vertical concatenation, where the diagram representing
the product $ab$ is obtained by placing the diagram representing $b$ atop the one representing $a$. 

\medskip

\noindent
\textbf{Remark.}
The PSG planar algebra `includes' the PSG algebras but not the other way around. 
A planar algebra $(A_n)_{n\in\mathbb{N}_0}$ offers a consistent way to define operations that are
not accessible to the individual algebras $A_n$ themselves, such as (unshaded versions of) 
the rotations \eqref{equ:ElRot} and traces \eqref{equ:TracePlanar}.
\medskip

We stress that there are no nontrivial relations involving $s_is_{i\pm 1}s_i$ without also involving terms with four or 
more $s$-factors. Accordingly, for $n=3$, there are infinitely many linearly independent vectors of the form
$(s_1s_2)^k$, hence $\dim \mathrm{PS}_3=\infty$, manifesting the non-evaluability of the planar algebra.

\section{Yang--Baxter integrable models}
\label{Sec:YBim}

We now review the integrability framework developed in \cite{PR22}, and introduce 
\textit{homogeneous transfer operators}, \textit{generalised Yang--Baxter equations}, and 
\textit{homogeneous Yang--Baxter integrability}. We also show that a singly generated planar algebra 
\textit{encoding} a homogeneous Yang--Baxter integrable model must 
(i) admit an unshaded description, and (ii) be YBR.

\subsection{Transfer operators}
\label{Sec:TYB}

Here, we recall the planar-algebraic Yang--Baxter integrability framework developed recently in \cite{PR22}.
Although we will be using terminology usually reserved for shaded planar algebras, including subfactor and YBR 
planar algebras, we will accordingly and henceforth work with unshaded versions of these planar algebras. 

For each $n\in\Nbb$, we follow \cite{PR22} and introduce the \textit{transfer tangle}
\be
 T_n:= \raisebox{-0.945cm}{\PlanarTemplateTransferAltRotRALTBNew}
\label{Tn}
\ee
and the parameterised algebra elements
\be
 K(u):= \sum_{a\in B_1} k_a(u)\,a =\!\raisebox{-0.45cm}{\RotKopBlueOp}\,,\quad
 R(u):= \sum_{a\in B_2} r_a(u)\,a = \!\raisebox{-0.345cm}{\RotRopGreenOpBBTLST}\,,\quad
 \Ko(u):= \sum_{a\in B_1} \ko_a(u)\,a = \!\raisebox{-0.45cm}{\RotKopYellowOp}\,,
\label{RKK}
\ee
where $B_n$ denotes a basis for $A_n$, while $k_a,r_a,\ko_a:\Omega\to\Cbb$, 
with $\Omega\subseteq\Cbb$ a suitable domain. We refer to $u$ parameterising the operators in \eqref{RKK},
as the corresponding \textit{spectral parameter}.
\medskip

\noindent 
\textbf{Remark.} The set $\Omega$ indicates a domain over which $R(u)$, $K(u)$, and $\Ko(u)$ are well-defined.
Typically, $\Omega$ contains an open set in $\mathbb{C}$, allowing power-series expansions of 
$R(u)$, $K(u)$, and $\Ko(u)$.
\medskip

\noindent
The associated \textit{homogeneous transfer operator} is defined as
\be
 T_n(u):=\Pb_{T_n}\big(K(u),R(u),\ldots,R(u),\Ko(u)\big),
\label{Tnu}
\ee
where $K(u)$ and $\Ko(u)$, respectively, are inserted into the left- and rightmost disk spaces in (\ref{Tnu}), 
while a copy of $R(u)$ is inserted into each of the remaining $2n$ disks, and we note that
\be
 T_n^{\dagger} = T_n,\qquad
 T_n(u)=\!\raisebox{-0.945cm}{\RotMarkRedTransferFilteredCrossIsoSingleRRotRALTBNew}\!\in A_n.
\label{Tnu2}
\ee
We refer to $R(u)$ as the corresponding \textit{$R$-operator}, and to $K(u)$ and $\Ko(u)$ as the 
corresponding \textit{$K$-operators}. For convenience, here we have adopted a slightly different 
convention compared to the definition of the transfer tangle in \cite{PR22}.
As in \cite{PR22}, the homogeneous transfer operator \eqref{Tnu} may be thought of as a double-row 
transfer operator on the strip, that is, a particular generalisation of a Sklyanin-type transfer matrix \cite{Sklyanin88}.
\medskip

\noindent
\textbf{Remark.}
More general transfer operators may be constructed, for example by including `inhomogeneities' 
at the level of the $R$-operator. \textit{Spectral inhomogeneities} are thus introduced by varying the spectral parameter 
of the $R$-operator depending on its position within the transfer tangle, while \textit{algebraic inhomogeneities} 
are introduced by varying the parameterisation in the construction of the $R$-operator (as an element of $A_2$) 
depending on its position within the transfer tangle, thereby introducing more than one $R$-operator. 
We refer to transfer operators with any of these features as \textit{inhomogeneous}. 
However, as we will exclusively consider homogeneous transfer operators, we often omit the qualifier ``homogeneous". 
\begin{prop}\label{prop:SelfAdRK}
If the $R$- and $K$-operators are self-adjoint with respect to the involution $\cdot^*$,
then so is the transfer operator $T_n(u)$ for each $n\in\Nbb$.
\end{prop}
\begin{proof}
Using \eqref{equ:InvNat}, we have
\begin{align}
        T_n(u)^* = \Pb_{T_n^{\dagger}}\big(K(u)^*,R(u)^*,\ldots,R(u)^*,\Ko(u)^*\big) =  T_n(u),
\end{align}
where the second equality follows from (\ref{Tnu2}) and the self-adjointness of the $R$- and $K$-operators. 
\end{proof}
\noindent
We denote the regular representation of $A_n$ by
\be
 \rho_n: A_n\to \mathrm{End}(A_n).
\ee
\begin{cor}\label{cor:genDiagFaith}
Let $(A_n)_{n\in\Nbb_0}$ be a subfactor planar algebra, and suppose the $R$- and $K$-operators 
are self-adjoint with respect to the trace form. Then, $\rho_n(T_n(u))$ is diagonalisable for all $n\in\Nbb$.
\end{cor}
\begin{proof}
Proposition \ref{prop:SelfAdRK} and the self-adjointness of the $R$- and $K$-operators imply that $T_n(u)$ is 
self-adjoint. By the spectral theorem, $\rho_n(T_n(u))$ is therefore diagonalisable. 
\end{proof}
With reference to the unshaded versions of the elementary rotation tangles in \eqref{equ:ElRot}, 
we say that the $R$- and $K$-operators (\ref{RKK}) are \textit{crossing symmetric} if
\be
    \Pb_{r_{1,1}}[K(u)] = \tilde{c}_K(u) K(c_K(u)),\quad 
    \Pb_{r_{2,1}}[R(u)] = \tilde{c}_R(u) R(c_R(u)),\quad
    \Pb_{r_{1,1}}[\Ko(u)] = \tilde{c}_{\Ko}(u) \Ko(c_{\Ko}(u)),
\ee
for some scalar functions $\tilde{c}_K,c_K,\tilde{c}_R,c_R,\tilde{c}_{\Ko},c_{\Ko}:\Omega\to\Cbb$.

\subsection{Yang--Baxter integrability}
\label{Sec:YBpol}

We say that a model described by the transfer operator $T_n(u)$ in (\ref{Tnu}) is \textit{integrable} on $\Omega$ if
\be
 [T_n(u),T_n(v)]=0,\qquad \forall\,u,v\in\Omega,
\label{TnTn}
\ee
with $\Omega\subseteq\Cbb$ a suitable domain. The following result is established in \cite{PR22} using standard
diagrammatic manipulations \cite{BPOB96}.

\begin{prop}
\label{Prop:YBE}
Let the parameterisations in (\ref{RKK}) be given, and suppose there exist
\begin{align}
  \raisebox{-0.345cm}{\RotRopDarkRedOpBBi} :=\sum_{a\in B_2} \yo^{(i)}_a(u,v)\,a,\qquad\quad
  &\raisebox{-0.345cm}{\RotRopRedOpBBi} :=\sum_{a\in B_2} y^{(i)}_a(u,v)\,a,
\label{XY}
\end{align}
where $\yo^{(i)}_a$ \!and\, $y^{(i)}_a$, $i=1,2,3$, are scalar functions defined for 
all $u,v\in\Omega\subseteq\Cbb$, such that the following three sets of relations are satisfied:
\begin{enumerate}
\item[$\bullet$] Inversion identities
\begin{align}
  \raisebox{-1.05cm}{\FramedInvLHS}\; =\! \raisebox{-1.05cm}{\FramedInvRHS}
  \qquad\quad (i=1,2,3)
\label{Invs}
\end{align}
\item[$\bullet$] Yang--Baxter equations
\be
 \raisebox{-1cm}{\FramedYBELHSI}\ 
 =\raisebox{-1cm}{\FramedYBERHSI}
\;\;\
 \raisebox{-1cm}{\FramedYBELHSII}\ 
 =\raisebox{-1cm}{\FramedYBERHSII} 
\;\;\
 \raisebox{-1cm}{\FramedYBELHSIII}\ 
 =\raisebox{-1cm}{\FramedYBERHSIII}
\label{YBE}
\ee
\item[$\bullet$] Boundary Yang--Baxter equations
\begin{align}
 \raisebox{-1.925cm}{\RotMarkRedBYBEbFilteredLeftALTBFramed}\ \;
 =\; \raisebox{-1.925cm}{\RotMarkRedBYBEbFilteredRightALTBFramed}
\qquad\qquad\qquad
 \raisebox{-1.925cm}{\RotMarkRedBYBEyFilteredLeftALTBFramed}\ \;
 =\; \raisebox{-1.925cm}{\RotMarkRedBYBEyFilteredRightALTBFramed}
\label{BYBEs}
\end{align}
where
\begin{align}
 \raisebox{-0.345cm}{\RotRopDarkRedOpBBid}\; =\sum_{a\in B_2} \yo^{(1)}_a(v,u)\,a,\qquad\quad
 &\raisebox{-0.345cm}{\RotRopRedOpBBid}\; =\sum_{a\in B_2} y^{(1)}_a(v,u)\,a.
\label{YBE4}
\end{align}
\end{enumerate}
Then, $[T_n(u),T_n(v)]=0$ for all $u,v\in\Omega$.
\end{prop}
\noindent
\textbf{Remark.}
As they appear in (\ref{YBE}), we label the equations YBE$_1$, YBE$_2$, and YBE$_3$.
\medskip

We denote the `auxiliary' operators in (\ref{XY}) by
\be
 \overline{Y\!}_i(u,v)=\sum_{a\in B_2} \yo^{(i)}_a(u,v)\,a,\qquad
 Y_i(u,v)=\sum_{a\in B_2} y^{(i)}_a(u,v)\,a,
\ee
and will refer to them as \textit{$Y$-operators}, as short for `YBE operators'. We refer to $\overline{Y\!}_i(u,v)$ as 
the \textit{spatial inverse} of $Y_i(u,v)$ if the pair satisfy (\ref{Invs}), 
and stress that pairs for different $i=1,2,3$ are, in general, independent of one another. Since the $Y$-operators need not be 
expressible in terms of the $R$-operator itself, we say that the inversion identities \eqref{Invs}, 
the YBEs (\ref{YBE}), and the boundary YBEs (\ref{BYBEs}) are \textit{generalised}.
We refer to YBE$_1$ under the specialisation $Y_1(u,v) = R(uv)$ as the \textit{standard} YBE.
\medskip

Traditionally, a model is referred to as \textit{Yang--Baxter integrable} if the $R$- and $K$-operators satisfy a set 
of local relations, including YBEs, that imply (\ref{TnTn}). The corresponding $R$- and $K$-operators are then 
said to provide a \textit{Baxterisation} \cite{Jones90}. Based on our homogeneous transfer operator (\ref{Tnu}), 
a prototypical set of such local relations is presented in Proposition \ref{Prop:YBE}. We accordingly refer to the 
ensuing integrability as \textit{homogeneous Yang--Baxter integrability}, and say that the corresponding $R$- and 
$K$-operators provide a \textit{homogeneous Baxterisation}. 
We view such a Baxterisation as \textit{specious} if
\be
 K(u)=k(u)a_1,\qquad R(u)=r(u)a_2,\qquad  \Ko(u)=\ko(u)\overline{a}_1,
\label{specious}
\ee
where $a_1,\overline{a}_1\in A_1$, $a_2\in A_2$ and $k,r,\ko:\Omega\to\Cbb$,
because, in that case, we have
\be
 T_n(u)=k(u)\ko(u)r^{2n}(u)a,\qquad a=\Pb_{T_n}\big(a_1,a_2,\ldots,a_2,\overline{a}_1\big)\in A_n,
\label{Tnukk}
\ee
from which (\ref{TnTn}) trivially follows. In the following, we will disregard specious Baxterisations. 
\medskip

\noindent
\textbf{Remark.}
As indicated in the Remark following \eqref{Tnu2}, a model may be described by a transfer operator with \textit{inhomogeneities}.
Adopting the approach in Proposition \ref{Prop:YBE}, one would then seek to use a modified set of sufficient conditions 
to establish the Yang--Baxter integrability of the model. Accordingly, such models are absent in 
the `homogeneous specialisation' discussed here.
\medskip

We say that a planar algebra $(A_n)_{n\in\Nbb_0}$ \textit{encodes} the homogeneous Yang--Baxter 
integrability if $\{K(u),$ $\overline{K}(u)\,|\,u\in\Omega\}$ and $\{R(u)\,|\,u\in\Omega\}$ together with the 
action of planar tangles generate the \textit{full} vector spaces $A_1$ and $A_2$, respectively, not just subspaces.
This notion will play a key role in subsequent sections.

As indicated, our focus will be on singly generated planar algebras, where $\dim A_1=|B_1|=1$.
We may accordingly normalise the $K$-operators in (\ref{RKK}) so that they equal the unit element,
\be
  K(u)=\Ko(u)=\mathds{1}_1.
\label{KK0}
\ee
We have opted to keep the exposition above general, but will be assuming that (\ref{KK0}) holds in the following.
In particular, the boundary YBEs (\ref{BYBEs}) now reduce to
\be
 \Pb_{r_{2,-1}}[\overline{Y}_{\!2}(u,v)]\overline{Y}_{\!1}(u,v)
  =\overline{Y}_{\!1}(v,u)\Pb_{r_{2,1}}[\overline{Y}_{\!3}(u,v)],\qquad
 \Pb_{r_{2,1}}[Y_2(u,v)]Y_1(u,v)=Y_1(v,u)\Pb_{r_{2,-1}}[Y_3(u,v)].
\label{bYBER}
\ee
\begin{prop}\label{Prop:YR}
Let $Y_1(u,v)=R(uv)$ solve $\mathrm{YBE}_1$.
Then,
\be
 Y_2(u,v)=\Pb_{r_{2,-1}}\big[R(\tfrac{u}{v})\big]\qquad\text{and}\qquad
 Y_3(u,v)=\Pb_{r_{2,1}}\big[R(\tfrac{u}{v})\big]
\label{YPR}
\ee
solve YBE$_2$ and YBE$_3$, respectively.
\end{prop}
\begin{proof}
By Corollary \ref{cor:ElRot} in Appendix \ref{app:pa}, $\Pb_{r_{n,k}}$ is invertible for all $n,k\in\Nbb$, 
so applying one of these maps to 
both sides of an equation yields an equivalent relation. If $Y_1(u,v)=R(uv)$, then applying $\Pb_{r_{3,-1}}$ 
and $\Pb_{r_{3,1}}$ to YBE$_2$ and YBE$_3$, respectively, yields conditions equivalent to YBE$_1$,
with $Y_2(u,v)=\Pb_{r_{2,-1}}[R(u/v)]$ and $Y_3(u,v)=\Pb_{r_{2,1}}[R(u/v)]$.
\end{proof}

\noindent
\textbf{Remark.}
As we will see in sections \ref{Sec:FC}--\ref{Sec:Liu}, the singly generated planar algebras admit Baxterisations, 
where the $Y$-operators are of the form presented in Proposition \ref{Prop:YR}.
While the present work thus does not exploit fully the generality of the YBEs \eqref{YBE}, we stress that, 
in the general setting, the generalised YBEs are inequivalent to the standard YBE.

\subsection{Unshaded planar algebras}
\label{Sec:Unshaded}

A key observation for us is that the singly generated planar algebras that \textit{do not} admit an unshaded 
description \textit{cannot} encode a homogeneous Yang--Baxter integrable model within the algebraic integrability 
framework outlined in Section \ref{Sec:TYB} and Section \ref{Sec:YBpol}. 

To see this, let $(A_{n,\pm})_{n\in\mathbb{N}_0}$ denote a singly generated planar algebra encoding a 
homogeneous Yang--Baxter integrable model, and consider the following linear maps that reverse the shading 
on the vectors in $A_{n,+}$ and $A_{n,-}$:
\begin{align}
    \iota_{n,+}:A_{n,+} \to A_{n,-}, \qquad 
    \raisebox{-0.505cm}{\IgPosPre}\,\mapsto \raisebox{-0.345cm}{\IgNegPre}\; ; \qquad\quad
    \iota_{n,-}:A_{n,-} \to A_{n,+}, \qquad 
    \raisebox{-0.345cm}{\IgNegPre}\,\mapsto \raisebox{-0.505cm}{\IgPosPre}\; ,
\end{align}
here illustrated for $n=2$.
Following \cite{LMP17}, there exists an unshaded planar algebra $(A_n)_{n\in\mathbb{N}_0}$ corresponding to 
$(A_{n,\pm})_{n\in\mathbb{N}_0}$ if and only if the map $\iota_{n,\mp}\circ \iota_{n,\pm}$ acts as the identity 
on $A_{n,\pm}$ for all $n\in\mathbb{N}_0$. 
Now, consider a shading consistent with the transfer operator \eqref{Tnu2},
\begin{align}
    T_n(u)=\!\raisebox{-0.945cm}{\includegraphics{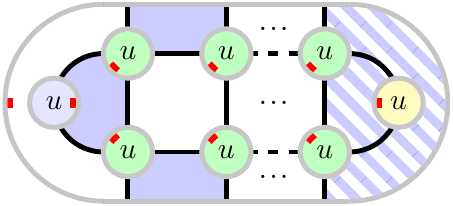}},
\end{align}
where 
\begin{align}
    R_+(u) :=  \raisebox{-0.505cm}{\RotRopGreenOpBBTLSTPos}\, \in A_{2,+}, \qquad\quad 
    R_-(u) := \raisebox{-0.345cm}{\RotRopGreenOpBBTLSTNeg}\in A_{2,-},
\end{align}
are shaded $R$-operators which, by homogeneity, satisfy 
\begin{align}\label{equ:iotaII}
    \iota_{2,\pm}[R_{\pm}(u)] = R_{\mp}(u), \qquad \forall\,u\in\Omega,
\end{align}
and the $K$-operators are shaded units, c.f.~\eqref{KK0}. As the planar algebra \textit{encodes} the integrability of the model, 
$\{R_{\pm}(u)\,|\,u\in\Omega\}$ together with the action of planar tangles, generates
the vector space $A_{2,\pm}$. Using \eqref{equ:iotaII}, it follows that $\iota_{2,\mp}\circ\iota_{2,\pm}$ acts as the 
identity on $A_{2,\pm}$, so the corresponding singly generated planar algebra $(A_{n,\pm})_{n\in\mathbb{N}_0}$
admits an unshaded description. 

For our homogeneous Yang--Baxter integrability purposes, it thus 
suffices to consider \textit{unshaded} planar algebras only. We stress that a shaded planar algebra not admitting 
an unshaded description could, in fact, encode the structure of an integrable model; however, the corresponding 
transfer operator would necessarily be \textit{inhomogenous}, see the Remark following (\ref{Tnu2}) and the 
comments made in Section \ref{Sec:Discussion}.

\subsection{Singly generated planar algebras}
\label{Sec:Baxt}

With reference to the notion of homogeneous Yang--Baxter integrability outlined above,
including the definition (\ref{Tnu}) of homogeneous transfer operators and non-specious Baxterisations, 
we have the following result.
\begin{prop}\label{prop:SGYBIntYBR}
    Singly generated planar algebras that encode homogeneous Yang--Baxter integrability are YBR.
\end{prop}
\begin{proof}
Since singly generated planar algebras that do not admit an unshaded description cannot encode a homogeneous 
Yang--Baxter integrable model within our algebraic integrability framework, 
Proposition \ref{prop:UnSGeqPSG} allows us to focus on the PSG planar algebra and its quotients.
It thus suffices to show that if one does not impose conditions turning the PSG planar algebra into a YBR 
planar algebra, then there exist no $R$- and $Y$-operators such that (i) the YBE
\begin{align}\label{equ:YBRSGsuff}
    \raisebox{-1cm}{\RotMarkRedYBEgoLeftSingleRRotRFramed}\; 
    = \raisebox{-1cm}{\RotMarkRedYBEgoRightSingleRRotRFramed}
\end{align}
is satisfied, and (ii) the $Y$-operator has a spacial inverse.

Relative to the $A_2$-basis $\{\mathds{1}_2, e, s\}$, we introduce
\begin{align}
    R(u) = r_{\mathds{1}}(u) \mathds{1}_2 + r_{e}(u) e + r_s(u) s,\qquad
    Y(u,v) = y_{\mathds{1}}(u,v) \mathds{1}_2 + y_{e}(u,v) e + y_s(u,v)s,
\end{align}
diagrammatically corresponding to
\begin{align}\label{equ:RYOpSGBax}
    \!\!\!\!\!\!\raisebox{-0.325cm}{\RotRopGreenOpBBTLSTAlt}\!\!\! 
    =  r_{\mathds{1}}(u)\!\!\!\raisebox{-0.325cm}{\RotRopTLxIdSTAlt}\!\!\! 
    + r_{e}(u)\!\!\!\raisebox{-0.325cm}{\RotRopTLxESTAlt}\!\!\!
    +r_s(u)\!\!\!\raisebox{-0.325cm}{\RotRopTLxBRSTAlt}\,, \qquad
    \raisebox{-0.345cm}{\RotRopRedOpBBOne}\!\!\! 
    =  y_{\mathds{1}}(u,v)\!\!\!\raisebox{-0.325cm}{\RotRopTLxIdSTAlt}\!\!\! 
    + y_{e}(u,v)\!\!\!\raisebox{-0.325cm}{\RotRopTLxESTAlt}\!\!\! 
    + y_s(u,v)\!\!\!\raisebox{-0.325cm}{\RotRopTLxBRSTAlt}\,,
\end{align}
where $r_{\mathds{1}},r_{e},r_s:\Omega\to\Cbb$ and $y_{\mathds{1}},y_{e},y_s:\Omega\times\Omega\to\Cbb$. 
\medskip

\noindent
\textbf{Remark.}
If $r_s$ is the zero function, then the attempted Baxterisation is, in fact, encoded by the Temperley--Lieb planar 
algebra, and not by a singly generated planar algebra, hence not of relevance here.
\medskip

From (\ref{equ:YBRSGsuff}), we get
\begin{align}
 0=&\raisebox{-1cm}{\RotMarkRedYBEgoRightSingleRRotRFramed}\; 
  - \raisebox{-1cm}{\RotMarkRedYBEgoLeftSingleRRotRFramed}
\nonumber
\end{align}
\vspace{-.3cm}
\begin{align}
 =&\,\Big\{\big[r_{\mathds{1}}(u) r_e(v)+r_e(u) r_{\mathds{1}}(v)+\delta\,r_e(u) r_e(v)
  -\tfrac{1}{\delta}r_s(u) r_s(v)\big]y_{\mathds{1}}(u,v)
\nonumber\\[.2cm]
    &\ -\big[r_{\mathds{1}}(u) r_{\mathds{1}}(v)-r_e(u) r_e(v)\big]y_e(u,v)
    -\tfrac{1}{\delta}\big[\epsilon\,r_e(u) r_s(v)
    +\tfrac{1}{\epsilon}r_s(u) r_e(v)\big]y_s(u,v)\Big\}(e_1-e_2)
\nonumber\\[.2cm]
 +&\,\Big\{\big[r_{\mathds{1}}(u) r_s(v)+r_s(u) r_{\mathds{1}}(v)+\alpha\,r_s(u) r_s(v)\big]y_{\mathds{1}}(u,v)
  -r_{\mathds{1}}(u) r_{\mathds{1}}(v) y_s(u,v)\Big\}
    (s_1-s_2)
\nonumber\\[.2cm]
 +&\,\Big\{\big[r_{\mathds{1}}(u) r_s(v)-\epsilon\,r_s(u)r_e(v)\big]y_e(u,v)-\big[r_{\mathds{1}}(u)r_e(v)
  +\alpha\,r_s(u)r_e(v)\big]y_s(u,v)\Big\}(s_1e_2 - e_1s_2)
\nonumber\\[.2cm]
 +&\,\Big\{\big[\tfrac{1}{\epsilon}r_e(u) r_s(v)-r_s(u) r_{\mathds{1}}(v)\big]y_e(u,v)
  +\big[r_e(u)r_{\mathds{1}}(v)+\alpha\,r_e(u)r_s(v)\big]y_s(u,v)\Big\}(s_2e_1-e_2s_1)
\nonumber\\[.2cm]
 +&\,\big\{r_s(u) r_s(v) y_s(u,v)\big\}(s_1s_2s_1-s_2s_1s_2),
\label{0proto}
\end{align}
where
\be
 e_1=\raisebox{-0.425cm}{\TLvecEI},\quad
 s_1=\raisebox{-0.425cm}{\TLvecSI},\quad
 e_2=\raisebox{-0.425cm}{\TLvecEII},\quad
 s_2=\raisebox{-0.425cm}{\TLvecSII}.
\ee
First, suppose $(s_1s_2s_1-s_2s_1s_2)$ is linearly independent of the other algebra elements in \eqref{0proto}.
Then, for the corresponding term to vanish, $r_s$ or $y_s$ must be zero, with the observation following 
(\ref{equ:RYOpSGBax}) subsequently implying that the function $y_s$ must be zero.
As $r_s\neq0$ and $y_s=0$, the vanishing of the $(s_2e_1-e_2s_1)$-term in (\ref{0proto}) implies that
\be
 \big[r_e(u)r_s(v)-\epsilon\,r_s(u)r_{\mathds{1}}(v)\big]y_e(u,v)=0.
\ee
Since $y_s=0$, the spacial invertibility of the $Y$-operator implies that $y_e\neq0$, so
\be
 \frac{r_e(u)}{r_s(u)}=\epsilon\left(\frac{r_{\mathds{1}}(v)}{r_s(v)}\right),
\ee
amounting to a $u,v$-independent constant. It follows that the $R$-operator is of the form in (\ref{specious}), so 
the Baxterisation is specious. 

Finally, if $(s_1s_2s_1-s_2s_1s_2)$ can be expressed as a linear combination of other terms in \eqref{0proto},
then we are in a quotient of the PSG planar algebra that is YBR. 
\end{proof}
In sections \ref{Sec:FC}--\ref{Sec:Liu}, we present concrete examples of YBR planar 
algebras encoding homogeneous Yang--Baxter integrability, thereby establishing the \textit{existence} of planar 
algebras satisfying the assumptions in Proposition \ref{prop:SGYBIntYBR}.
Indeed, we use Liu's classification in Theorem \ref{thm:LiuYBRClass} below to show that, for \textit{every}
singly generated YBR planar algebra, there exists a corresponding Yang--Baxter integrable model.
\begin{thm}[Liu \cite{Liu15}]\label{thm:LiuYBRClass}
A singly generated YBR planar algebra is isomorphic to a quotient of an FC, BMW, or Liu planar algebra. 
\end{thm}
\noindent
The planar algebras listed in Theorem \ref{thm:LiuYBRClass} are recalled in 
Section \ref{Sec:FC}, Section \ref{Sec:BMW}, and Section \ref{Sec:Liu}, respectively.
In Section \ref{Sec:pol}, we show that these models are not only Yang--Baxter integrable but also
\textit{polynomially integrable}. The paragraph including Theorem \ref{thm:1} and Theorem \ref{cor:1}
in Section \ref{Sec:Intro} summarises these key findings.
\medskip

\noindent
\textbf{Remark.}
If the FC, BMW, or Liu planar algebra in Theorem \ref{thm:LiuYBRClass} is positive semi-definite, then its quotient 
by the kernel of the trace norm (that is, the quotient by the ideal generated by all nonzero $v$ for which 
$\Pb_{\tr_n}(v^*v)=0$ for some $n$) is isomorphic to the corresponding YBR planar algebra. 
However, to keep $\dim A_2(S,C)=3$ in (\ref{AAA}),
we will only apply this quotient operation for $n>2$, c.f.~the Remark following (\ref{Un}).

\section{Fuss--Catalan algebra}
\label{Sec:FC}

\subsection{Planar algebra}
\label{subsec:FCDaP}

Let $F_{n,\pm}^{(k)}$ denote the vector space spanned by disks with $2kn$ nodes such that
(i) each node is labelled by one of the $k$ colours $c_1,\ldots,c_k$, 
(ii) clockwise from the marked interval, nodes are assigned colours according to
\begin{align}
\underbrace{(c_1,\ldots,c_k),(c_k,\ldots,c_1),(c_1,\ldots,c_k),\ldots,(c_k,\ldots,c_1)}_{\#(\ldots)\,
 =\,2n} \qquad\ \text{for disks in $F_{n,+}^{(k)}$},
\\[.1cm]
\underbrace{(c_k,\ldots,c_1),(c_1,\ldots,c_k),(c_k,\ldots,c_1),\ldots,(c_1,\ldots,c_k)}_{\#(\ldots)\,
 =\,2n} \qquad\ \text{for disks in $F_{n,-}^{(k)}$},
\end{align}
and (iii) every node is connected to another node with the same colour, using non-intersecting strings 
defined up to ambient isotopy. Examples of such \textit{Fuss--Catalan disks} are
\begin{align}
    \raisebox{-1cm}{\FCExVecI},\raisebox{-1cm}{\FCExVecII} \!\in F_{2,+}^{(3)},\qquad
    \raisebox{-1cm}{\FCExVecIIINeg},\raisebox{-1cm}{\FCExVecIVNeg} \!\in F_{3,-}^{(2)}.
\end{align}
The vector-space dimensions are given by Fuss--Catalan numbers, as
\begin{align}
 \dim F_{n,\pm}^{(k)} =\frac{1}{kn+1}\binom{kn+n}{n}.
\end{align}

The \textit{$k$-coloured Fuss--Catalan planar algebra} $\mathrm{FC}^{(k)}(\gamma_1,\ldots,\gamma_k)$ is the 
vector-space collection $(F_{n,\pm}^{(k)})_{n\in\Nbb_0}$ together with the following action of shaded planar 
tangles \cite{JonesNotes}: 
(i) for each string within a shaded planar tangle, draw $k-1$ parallel strings in the adjacent unshaded region 
and assign each a label $c_k\ldots,c_1$ starting from the original string 
(for $k>1$, the tangle shading is thus encoded in the string labels and can thereafter be omitted), 
(ii) if a loop is formed with the colour $c_l$, then it is removed and replaced by the scalar weight $\gamma_l$, 
and (iii) the output vector is given by the output disk with the given colour labels and ensuing
string connections. To illustrate, we have
\begin{align}
    T =  \raisebox{-1.4275cm}{\FCExTangleAltShaded}\,,\qquad
    \Pb_T\big(\,\raisebox{-0.3cm}{\rotatebox{90}{\RotRopFCxESTBlackAltColoured}}\,,\!
    \raisebox{-0.175cm}{\RotRopFCxXSTBlackAltColoured}\,\big) 
    \,= \raisebox{-1.4275cm}{\FCExTangleInputAltColoured} 
    \,=\, \gamma_1\raisebox{-0.43cm}{\FCExTangleInputIdColoured}\;,
\end{align}
where $c_1$ and $c_2$ correspond to  the colours cyan and black, respectively. 
If $\gamma_1=\cdots=\gamma_m$, then the colours of the strings are immaterial and the 
corresponding planar algebra admits an unshaded description. 

As we are concerned with unshaded singly generated planar algebras, we denote by $\mathrm{FC}(\gamma)$
the unshaded planar algebra corresponding to $\mathrm{FC}^{(2)}(\gamma,\gamma)$, and refer to it simply as 
\textit{Fuss--Catalan} (FC). From \cite{BJ97}, we know that $\mathrm{FC}(\gamma)$ is spherical and involutive, 
with the involution $\cdot^*$ defined as the conjugate linear map that acts by reflecting every disk about a line 
perpendicular to its marked boundary interval.

Let $\mathcal{F}$ denote the set of all $\gamma$ such that $\mathrm{FC}(\gamma)$ is positive semi-definite. 
For each $\gamma\in\mathcal{F}$, the \textit{FC subfactor planar algebra} $\mathsf{FC}(\gamma)$ is then defined 
as the quotient of $\mathrm{FC}(\gamma)$ by the kernel of the trace norm. Details of the set $\mathcal{F}$ are 
presented in \cite{BJ97}, including $\big\{2\cos\tfrac{\pi}{m+2}\,|\, m\in\mathbb{N}\big\} \subset \mathcal{F}$.
For $\gamma$ in that discrete subset, $\mathrm{FC}(\gamma)$ is positive semi-definite, while for $\gamma>2$, 
$\mathrm{FC}(\gamma)$ is positive-definite, in which case $\mathrm{FC}(\gamma) \cong \mathsf{FC}(\gamma)$.

\subsection{Presentation}

For each $n\in\Nbb$, the \textit{FC algebra} $\mathrm{FC}_{n}(\gamma)$ is now defined by endowing the 
(unshaded) vector space $F_{n}^{(2)}$ with the multiplication induced by the unshaded planar tangle $M_n$ 
following from \eqref{equ:MultAndCoMultPlanar}. We note that the FC algebra is both unital (with unit denoted by 
$\mathds{1}$) and associative, and that it is a $^*$-algebra with involution inherited from the FC planar algebra. 
Using a diagrammatic representation similar to the one in (\ref{1es}),
$\mathrm{FC}_{n}(\gamma)$ is generated by the following algebra elements:
\begin{align}
    \mathds{1}\leftrightarrow\raisebox{-0.9cm}{\FCRecIdColoured},\qquad
    P_i\leftrightarrow\raisebox{-0.91cm}{\FCRecXOddColoured},\qquad
    E_i\leftrightarrow\raisebox{-0.91cm}{\FCRecEOddColoured},
\end{align}
where each label below a diagram labels a \textit{pair} of string endpoints.

For $\gamma\neq0$, the FC algebra admits \cite{BJ97} a presentation,
\be
 \mathrm{FC}_{n}(\gamma)\cong\langle E_i,P_i\,|\,i=1,\ldots,n-1\rangle,
\ee
with relations
\be
\begin{array}{rllll}
  & E_i^2=\gamma^2 E_i,\quad\
  & P_iE_i= E_iP_i = \gamma E_i,\quad\
  & P_i^2=\gamma P_i,
  &
\\[.15cm]
  & E_iE_{i\pm1}E_i= E_i,\quad\
  & P_iE_{i\pm1}P_i=P_iP_{i\pm1}=P_{i\pm1}P_i,\quad\
  & E_iP_{i\pm1}E_i=\gamma E_i,
  &
\\[.15cm]
  & E_iE_j=E_jE_i,\quad\
  & E_iP_j=P_jE_i,\quad\
  & P_iP_j=P_jP_i,\qquad\
  & |i-j|>1.
\end{array}
\label{FCrel}
\ee
Following from (\ref{FCrel}), we also have the relations
\be
 P_iP_{i\pm1}E_i=\gamma P_{i\pm1}E_i,\qquad
 E_iP_{i\pm1}P_i=\gamma E_iP_{i\pm1},\qquad
 P_iP_{i\pm1}P_i=\gamma P_iP_{i\pm1},
\label{FCrel1}
\ee
and
\be
 E_iE_{i\pm1}P_i=E_iP_{i\pm1},\qquad
 P_iE_{i\pm1}E_i=P_{i\pm1}E_i.
\label{FCrel2}
\ee
For $\gamma=0$, the relations (\ref{FCrel}), (\ref{FCrel1}), and (\ref{FCrel2}) still hold, but the 
relations in (\ref{FCrel2}) do not all follow from the relations in (\ref{FCrel}), and should be imposed separately.
\medskip

\noindent
\textbf{Remark.}
The Temperley--Lieb subalgebra $\langle E_1,\ldots,E_{n-1}\rangle\subset\mathrm{FC}_{n}(\gamma)$ 
has loop fugacity $\delta=\gamma^2$.
\medskip

We let $\mathcal{F}_n$ denote the set of all $\gamma$ such that the trace form \eqref{form} is positive 
semi-definite on $F_{n}^{(2)}$, noting that $\mathcal{F}\subseteq\mathcal{F}_n$ for all $n\in\mathbb{N}_0$.  
For each $\gamma\in\mathcal{F}_n$, $\mathsf{FC}_{n}(\gamma)$ is then defined 
as the quotient of $\mathrm{FC}_{n}(\gamma)$ by the kernel of the trace norm.

\subsection{Quotient description}
\label{subsec:FCdiag}

\begin{prop}\label{Prop:FCq}
For $\gamma^2 > 1$ and each $\mu\in\{-1,1\}$, we have
\be
 \mathrm{FC}_n(\gamma)\cong\mathrm{PS}_n^{(1)}(\mu(\gamma-\gamma^{-1}),\gamma^2)
  \big/\big\langle\iota_{i,j}\,|\,|i-j|=1;\,i,j\in\{1,\ldots,n-1\}\big\rangle,
\ee
where
\be
 \iota_{i,j}=\sh_ie_j\sh_i-\sh_i\sh_j+\sh_je_i+e_j\sh_i+e_je_i-\sh_i-\sh_j-\mathds{1},
\label{iotash}
\ee
with $\sh_i=\mu\gamma s_i$.
\end{prop}
\begin{proof}
The proposed algebra isomorphism sets
\be
 E_i=e_i,\qquad 
 P_i=\frac{1}{\gamma+\gamma^{-1}}\big(\mathds{1}+e_i+\sh_i\big),
\ee
or equivalently,
\be
 e_i=E_i,\qquad \sh_i=-\mathds{1}-E_i+(\gamma+\gamma^{-1})P_i.
\ee
With this, one verifies that the relations (\ref{FCrel}) imply the relations (\ref{PSrel}) and the vanishing of 
(\ref{iotash}), and that the relations (\ref{PSrel}) together with the vanishing of (\ref{iotash}) imply the relations 
(\ref{FCrel}).
\end{proof}
\noindent
\textbf{Remark.}
By renormalising the $P_i$ generators, introducing $\Ph_i:=P_i/\gamma$, the relations 
(\ref{FCrel})--(\ref{FCrel2}) only depend on the loop weight through $\delta=\gamma^2$, as we then have
\be
\begin{array}{rllll}
  & E_i^2=\delta E_i,\quad\ 
  & \Ph_iE_i= E_i\Ph_i =E_i,\quad\
  & \Ph_i^2=\Ph_i,
  &
\\[.15cm]
  & E_iE_{i\pm1}E_i= E_i,\quad\
  & \Ph_iE_{i\pm1}\Ph_i=\Ph_i\Ph_{i\pm1}=\Ph_{i\pm1}\Ph_i,\quad\
  & E_i\Ph_{i\pm1}E_i=E_i,\quad\
  &
\\[.15cm]
  & E_iE_j=E_jE_i,\quad\
  & E_i\Ph_j=\Ph_jE_i,\quad\
  & \Ph_i\Ph_j=\Ph_j\Ph_i,\qquad\
  & |i-j|>1,
\end{array}
\ee
and
\be
 \Ph_i\Ph_{i\pm1}E_i=\Ph_{i\pm1}E_i,\qquad
 E_i\Ph_{i\pm1}\Ph_i=E_i\Ph_{i\pm1},\qquad
 \Ph_i\Ph_{i\pm1}\Ph_i=\Ph_i\Ph_{i\pm1},
\ee
\be
 E_iE_{i\pm1}\Ph_i=E_i\Ph_{i\pm1},\qquad
 \Ph_iE_{i\pm1}E_i=\Ph_{i\pm1}E_i.
\ee

\subsection{Baxterisation}

Relative to the canonical $F_{2}^{(2)}$-basis $\{\mathds{1}_2,E,P\}$, 
we introduce the parameterised $R$-operator as
\begin{align}
\label{RKFC}
 R(u) = r_{\mathds{1}}(u)\mathds{1}_2  + r_{E}(u)E + r_P(u)P,\qquad
    \raisebox{-0.3275cm}{\RotRopGreenOpBBFCSTBlack}\,
    =r_{\mathds{1}}(u)\,\raisebox{-0.325cm}{\RotRopFCxIdSTBlack}\,
    + r_{E}(u)\,\raisebox{-0.435cm}{\rotatebox{90}{\RotRopFCxESTBlack}}\,
    + r_{P}(u)\raisebox{-0.325cm}{\RotRopFCxXSTBlack}\,,
\end{align}
with $r_{\mathds{1}},r_E,r_P:\Omega\to\Cbb$. We note that
\be
 \Pb_{r_{2,\pm1}}[R(u)]= r_{E}(u)\mathds{1}_2  +r_{\mathds{1}}(u)E + r_P(u)P,\qquad
 \Pb_{r_{2,2}}[R(u)]=R(u).
\ee

\noindent
\textbf{Remark.}
Although $\langle E_1,\ldots,E_{n-1}\rangle$ and $\langle P_1,\ldots,P_{n-1}\rangle$ are subalgebras
of $\mathrm{FC}_{n}(\gamma)$, the $P$-generators do not form a \textit{planar} subalgebra of 
$\mathrm{FC}(\gamma)$. It follows from the Remark after \eqref{equ:RYOpSGBax} that only $r_P$ is required 
to be nonzero when exploring homogeneous integrability encoded by $\mathrm{FC}_{n}(\gamma)$.
\medskip

It is known \cite{DiF98} that $\mathrm{FC}_{n}(\gamma)$
admits a Baxterisation. In the notation (\ref{RKFC}), $r_{\mathds{1}}$ is nonzero, so $R$ may be normalised 
such that $r_{\mathds{1}}(u)=1$. Using $\Ph=P/\gamma$ and
\be
 f:x\mapsto\frac{x(x-1)}{\delta-1-x}, 
\ee
the homogeneous Baxterisation then reads
\be\label{equ:FCRopParam}
  R(u)=\mathds{1}_2+f(u) E+(u-1)\Ph,
\ee
with
\be
 Y\!_1(u,v) = R(uv), \qquad 
 \overline{Y\!}_1(u,v)=f\big(\tfrac{uv}{\delta-1}\big)f\big(\tfrac{\delta-1}{uv}\big) R\big(\tfrac{(\delta-1)^2}{uv}\big), \label{Y1FC}
\ee
and
\be
  Y\!_2(u,v)=Y\!_3(u,v)=\Pb_{r_{2,\pm1}}\big[R\big(\tfrac{u}{v}\big)\big],\qquad
  \overline{Y\!}_2(u,v)=\overline{Y\!}_3(u,v)=\Pb_{r_{2,\pm1}}\big[R\big(\tfrac{v}{u}\big)\big].
\ee
Using the crossing symmetry
\begin{align}\label{equ:FCCrossSym}
    \Pb_{r_{2,\pm1}}[R(u)]  = f(u)R\big(\tfrac{\delta-1}{u}\big), 
\end{align}
the $Y$-operators can be expressed explicitly in terms of $R$-operators, and the conditions \eqref{Invs} 
and \eqref{YBE} reduce to the following single inversion identity and single (\textit{standard}) YBE:
\begin{align}
  \raisebox{-1.055cm}{\RotMarkRedInvrFilteredregwOrSpecFCFramed}\;
  =\raisebox{-1.055cm}{\MarkRedIdFilteredInvFCFramed}\ , \qquad\quad 
    \raisebox{-1cm}{\RotMarkRedYBEgoLeftSingleRregOrFCFramed}\ \,
    =\raisebox{-1cm}{\RotMarkRedYBEgoRightSingleRregOrFCFramed}\ .
\end{align}
Relative to the involution $\cdot^*$ on the FC algebra, the $R$-operator \eqref{equ:FCRopParam} is self-adjoint 
for all $u\in\mathbb{R}\setminus\{\delta - 1\}$.
\medskip

\noindent
\textbf{Remark.}
We have verified that, up to a normalising factor, the generalised Yang--Baxter framework of 
Proposition \ref{Prop:YBE} does not admit any other non-specious solution of the form (\ref{RKFC}), with $r_P$ 
nonzero, than the one presented above.

\section{Birman--Wenzl--Murakami algebra}
\label{Sec:BMW}

\subsection{Planar algebra}
\label{subsec:BMWDaP}

Let $W_n$ denote the vector space spanned by disks with $2n$ nodes such that, within each disk,
(i) the nodes are connected pairwise by strings, defined up to regular isotopy, 
(ii) strings may intersect but not self-intersect, 
(iii) two strings cannot intersect one another more than once, 
and (iv) strings cannot form loops. 
To illustrate, we present the following examples and non-examples:
\begin{align}\label{equ:BMWExAndNEx}
    \raisebox{-0.45cm}{\BMWExVecAcI},
    \raisebox{-0.45cm}{\BMWExVecAcII},
    \raisebox{-0.45cm}{\BMWExVecAcIII}\qquad\mathrm{and}\qquad
    \raisebox{-0.45cm}{\BMWExVecI},
    \raisebox{-0.45cm}{\BMWExVecII},
    \raisebox{-0.45cm}{\BMWExVecIII}.
\end{align}
The dimension of $W_n$ is given by
\be
 \dim W_n=(2n-1)!!.
\ee

For each pair of scalars $\tau\neq0$ and $q\notin\{-1,0,1\}$, the \textit{Birman--Wenzl--Murakami} (BMW) 
\textit{planar algebra} $\mathrm{BMW}(\tau,q)$ is the collection of vector spaces $(W_n)_{n\in\Nbb_0}$ together 
with the natural diagrammatic action of planar tangles, defined up to regular isotopy and 
subject to the relations
\be
 \raisebox{-0.285cm}{\BMWloopAlt} 
    =\delta\raisebox{-0.285cm}{\BMWemptyAlt},\qquad 
    \raisebox{-0.285cm}{\BMWtwistIIAlt} = \tau\raisebox{-0.285cm}{\BMWidAlt},\qquad 
    \raisebox{-0.285cm}{\BMWtwistIAlt} = \tau^{-1}\raisebox{-0.285cm}{\BMWidAlt}, 
\label{BMWv}
\ee
\be
 \raisebox{-0.285cm}{\BMWbraidAlt} - \raisebox{-0.285cm}{\BMWbraidInvAlt} 
 = Q\left[\!\raisebox{-0.285cm}{\BMWidIIAlt} - \raisebox{-0.285cm}{\BMWeProjAlt}\!\right],
\ee
where
\be
 \delta = 1+\frac{\tau-\tau^{-1}}{Q},\qquad Q=q-q^{-1}.
\label{BMWdelta}
\ee
\noindent
\textbf{Remark.}
Self-intersecting or loop-forming strings may arise as the result of a planar tangle acting on
vectors in $(W_n)_{n\in\Nbb_0}$, hence the relevance of relations like (\ref{BMWv}).
\medskip

The planar algebra $\mathrm{BMW}(\tau,q)$ is spherical and, for $|\tau|=|q|=1$ or 
$\tau,q\in\mathbb{R}$, involutive \cite{BJL17, Wenzl90}. In these cases, 
the involution $\cdot^*$ is defined as the conjugate linear map that acts by `reflecting' respectively 
`flipping' every disk about a line perpendicular to its marked boundary interval, as indicated by
\begin{align}
    \left(\!\!\raisebox{-0.285cm}{\BMWidIIAlt}\!\!\right)^* = \raisebox{-0.285cm}{\BMWidIIAlt}, \qquad\quad 
    \left(\!\!\raisebox{-0.285cm}{\BMWeProjAlt}\!\!\right)^* = \raisebox{-0.285cm}{\BMWeProjAlt}, \qquad\quad 
    \left(\!\!\raisebox{-0.285cm}{\BMWbraidAlt}\!\!\right)^* = 
    \begin{cases} \raisebox{-0.285cm}{\BMWbraidInvAlt},\ \ & |\tau|=|q|=1, \\[.3cm]
       \raisebox{-0.285cm}{\BMWbraidAlt},\ \ & \tau,q\in\Rbb,
    \end{cases}
\end{align}
recalling that $q\neq\pm1$.

We let $\mathcal{B}$ denote the set of all $(\tau,q)$ such that $\mathrm{BMW}(\tau,q)$ is 
positive semi-definite. For each $(\tau,q)\in\mathcal{B}$, the \textit{BMW subfactor planar algebra} 
$\mathsf{BMW}(\tau,q)$ is then defined as the quotient of $\mathrm{BMW}(\tau,q)$ by the kernel of the trace 
norm. Details of the set $\mathcal{B}$ are presented in \cite{BJL17}, including the observation that $\mathcal{B}\neq\emptyset$.

\subsection{Presentation}
\label{Sec:BMWpresentation}

For each $n\in\Nbb$, the \textit{BMW algebra} $\mathrm{BMW}_{n}(\tau,q)$ is defined by 
endowing the vector space $W_n$ with the multiplication induced by the unshaded planar tangle $M_n$ 
following from \eqref{equ:MultAndCoMultPlanar}. We note that the BMW algebra is both unital 
(with unit denoted by $\mathds{1}$) and associative, and that, for $|\tau| = |q| = 1$ or $\tau,q\in\mathbb{R}$, 
it is a $^*$-algebra with involution inherited from the BMW planar algebra. 
As is well-known \cite{Kauffman87,MW89,Kauffman90}, 
the generators of $\mathrm{BMW}_{n}(\tau,q)$ can be represented diagrammatically as
\begin{align}\label{BMWdiag}
    \mathds{1}\leftrightarrow \raisebox{-0.885cm}{\BMWRecIdOdd},\quad 
    g_i\leftrightarrow\raisebox{-0.885cm}{\BMWRecBOdd},\quad
    g_i^{-1}\leftrightarrow\raisebox{-0.885cm}{\BMWRecBInvOdd},\quad
    e_i\leftrightarrow\raisebox{-0.885cm}{\BMWRecEOdd}.
\end{align}

The algebra $\mathrm{BMW}_{n}(\tau,q)$ admits \cite{Wenzl90} a presentation,
\be
 \mathrm{BMW}_{n}(\tau,q)\cong\langle e_i,g_i,g_i^{-1}\,|\,i=1,\ldots,n-1\rangle,
\ee
with relations
\be
\begin{array}{rllll}
  & g_ig_{i\pm1}g_i= g_{i\pm1}g_ig_{i\pm1},\qquad
  & g_i-g_i^{-1}=Q(\mathds{1}-e_i),
  &
\\[.15cm]
  & g_{i}e_{i\pm1}g_{i} = g_{i\pm1}^{-1}e_{i}g_{i\pm1}^{-1},\qquad
  & g_ie_i = e_ig_i = \tau^{-1}e_i,
  &
\\[.15cm]
  & e_ig_{i\pm1}g_i= g_{i\pm1}g_ie_{i\pm1}= e_{i}e_{i\pm1},\qquad
  & g_ig_j=g_jg_i,\qquad
  & |i-j|>1.
\end{array}
\label{BMWrel}
\ee
It follows from these relations that
\be
 e_i^2=\delta e_i,\qquad
 e_ie_{i\pm1}e_i= e_i,\qquad
 g_ie_j=e_jg_i,\qquad
 e_ie_j=e_je_i,\qquad
 |i-j|>1,
\ee
with $\delta$ as in (\ref{BMWdelta}), and that
\be
 g_{i}e_{i\pm 1}e_{i} = g_{i\pm 1}^{-1}e_{i},\qquad
    e_{i}e_{i\pm 1}g_{i} = e_{i}g_{i\pm 1}^{-1},\qquad
    e_{i}g_{i\pm 1}e_{i} = \tau e_{i}.
\ee
We note that it suffices to list one of the two relations 
$g_ie_i=\tau^{-1}e_i$ and $e_ig_i=\tau^{-1}e_i$ in (\ref{BMWrel}).

We let $\mathcal{B}_n$ denote the set of all $(\tau,q)$ such that the trace form \eqref{form} is positive 
semi-definite on $W_{n}$, noting that $\mathcal{B}\subseteq\mathcal{B}_n$ for all $n\in\mathbb{N}_0$. 
For each $(\tau,q)\in\mathcal{B}_n$, $\mathsf{BMW}_n(\tau,q)$ is then defined as the quotient 
of $\mathrm{BMW}_n(\tau,q)$ by the kernel of the trace norm.

\subsection{Quotient description}
\label{Sec:BMWquotient}

Let
\be
 \Gamma:=(\tau^2+Q\tau-1)(\tau^2+Q(Q^2+3)\tau-1)=(\tau+q)(\tau-q^{-1})(\tau+q^3)(\tau-q^{-3}).
\ee
\begin{prop}\label{Prop:BMWq}
For $\delta > 1$, with $\delta$ parameterised as in (\ref{BMWdelta}), and for each $\mu\in\{-1,1\}$, we have
\be
 \mathrm{BMW}_n(\tau,q)\cong\mathrm{PS}_n^{(1)}
  \Big(\frac{\mu Q(\tau^2+1)}{\sqrt{\Gamma}},1+\frac{\tau-\tau^{-1}}{Q}\Big)
  \big/\big\langle\iota_1.\ldots,\iota_{n-2}\big\rangle,
\ee
where
\begin{align}
 \iota_i=\sh_i\sh_{i+1}\sh_i-\sh_{i+1}\sh_i\sh_{i+1}
 +&\,Q^2\tau\big\{(\tau^2+1)(\tau^2+Q(Q^2+3)\tau-1)[e_i-e_{i+1}]
 \nonumber\\[.15cm]
  &-(Q+\tau)(1-Q\tau)[\sh_i-\sh_{i+1}+e_i\sh_{i+1}-e_{i+1}\sh_i+\sh_{i+1}e_i-\sh_ie_{i+1}]\big\}
\label{iotagh}
\end{align}
with $\sh_i=\mu\sqrt{\Gamma}s_i$.
\end{prop}
\begin{proof}
The proposed algebra isomorphism uses the same notation, $e_i$, for the Temperley--Lieb generators, and sets
\be
 s_i=\frac{\mu}{\sqrt{\Gamma}}\big[Q\tau(Q+\tau)\mathds{1}+Q(1-Q\tau)e_i-(\tau^2+2Q\tau-1)g_i\big],
\label{sto}
\ee
or equivalently,
\be
 g_i=\frac{1}{\tau^2+2Q\tau-1}\big[Q\tau(Q+\tau)\mathds{1}+Q(1-Q\tau)e_i-\mu\sqrt{\Gamma}s_i\big].
\ee
With this, one verifies that the relations (\ref{BMWrel}) imply the relations (\ref{PSrel}) and the vanishing of 
(\ref{iotagh}). Likewise, the relations (\ref{PSrel}) together with the vanishing of (\ref{iotagh}) are seen
to imply the relations (\ref{BMWrel}).
\end{proof}
\noindent
\textbf{Remark.}
For $\epsilon=1$, $s\in W_2$ is invariant under $\Pb_{r_{2,1}}$, 
as becomes evident when rewriting (\ref{sto}) as
\be
 s_i=\frac{\mu Q(\tau^2+1)}{2\sqrt{\Gamma}}\,(\mathds{1}+e_i)
   -\frac{\mu(\tau^2+2Q\tau-1)}{2\sqrt{\Gamma}}\,(g_i+g_i^{-1}),
\ee
since
\be
 \Pb_{r_{2,1}}(\mathds{1}_2)=e,\qquad
 \Pb_{r_{2,1}}(e)=\mathds{1}_2,\qquad
 \Pb_{r_{2,1}}(g^{\pm1})=g^{\mp1}.
\ee

\noindent
\textbf{Remark.}
The algebra $\mathrm{BMW}_{n}(\tau,q)$ is occasionally referred to as \textit{Kauffman's Dubrovnik version}
\cite{Kauffman90}. It differs from the one in \cite{BW89}, which is based on
\be
 g_i+g_i^{-1} = Q_0(\mathds{1}+e_i),\qquad
 e_i^2=\delta_0 e_i,\qquad 
 \delta_0=-1+\frac{\tau+\tau^{-1}}{Q_0},\qquad
 Q_0=q+q^{-1},
\ee
and consequently admits a description as a quotient of $\mathrm{PS}_n^{(\epsilon)}$ similar to the one in 
Proposition \ref{Prop:BMWq}, but for $\epsilon=-1$. The two possible imaginary values, $\epsilon=\mathrm{i}$ and 
$\epsilon=-\mathrm{i}$, enter similar quotient descriptions of the Liu algebra in Section \ref{Sec:Liupresentation}.

\subsection{Baxterisation}

Relative to the canonical $W_2$-basis $\{\mathds{1}_2, e, g\}$, we introduce the parameterised $R$-operator as
\begin{align}\label{RegK}
 R(u) = r_{\mathds{1}}(u)\mathds{1}_2 + r_{e}(u)e + r_g(u)g,\qquad 
 \raisebox{-0.325cm}{\RotRopGreenOpBBTLSTAlt} 
     =  r_{\mathds{1}}(u) \raisebox{-0.325cm}{\RotRopTLxIdSTAlt} 
     + r_{e}(u)\raisebox{-0.325cm}{\RotRopTLxESTAlt} 
     + r_g(u)\raisebox{-0.325cm}{\RotRopTLxBRSTAltG}\;,
\end{align}
with $r_{\mathds{1}},r_e,r_g:\Omega\to\Cbb$. 
We note that
\be
 \Pb_{r_{2,\pm1}}[R(u)]= \big[r_{e}(u)-Qr_g(u)\big]\mathds{1}_2 + \big[r_{\mathds{1}}(u)+Qr_g(u)\big]e + r_g(u)g,\qquad
 \Pb_{r_{2,2}}[R(u)]=R(u).
\ee

\noindent
\textbf{Remark.}
Since $e_i$ is quadratic in $g_i$, the function $r_g$ is required to be nonzero when exploring homogeneous 
integrability encoded by $\mathrm{BMW}_n(\tau,q)$, see Remark following \eqref{equ:RYOpSGBax}.
\medskip

It is known \cite{CGX91}, see also \cite{PRT14}, that $\mathrm{BMW}_n(\tau,q)$ admits a Baxterisation. 
For each $\omega\in\{-\tau q,\tau q^{-1}\}$, using
\be
 b:x\mapsto\frac{\omega(1-x)(q^2-\omega/x)}{x(x-\omega)(q^2-x)},
\ee
the homogeneous Baxterisation then reads
\be\label{equ:BMWRopParam}
  R(u)=\frac{q^2-1}{q^2-u}\Big[\mathds{1}_2+\frac{1-u}{u-\omega}\,e+\frac{1-u}{Qu}\,g\Big],
\ee
with
\be
 Y\!_1(u,v) = R(uv), \qquad 
 \overline{Y\!}_1(u,v)
  =b\big(\tfrac{uv}{\omega}\big)b\big(\tfrac{\omega}{uv}\big)
  R\big(\tfrac{\omega^2}{uv}\big),
\label{Y1BMW}
\ee
and
\be
  Y\!_2(u,v)=Y\!_3(u,v)=\Pb_{r_{2,\pm1}}\big[R\big(\tfrac{u}{v}\big)\big],\qquad
  \overline{Y\!}_2(u,v)=\overline{Y\!}_3(u,v)=\Pb_{r_{2,\pm1}}\big[R\big(\tfrac{v}{u}\big)\big].
\ee
Using the crossing symmetry
\begin{align}\label{equ:BMWCrossSym}
    \Pb_{r_{2,\pm1}}[R(u)]  = b(u)R\big(\tfrac{\omega}{u}\big), 
 \end{align}
the $Y$-operators can be expressed explicitly in terms of $R$-operators, and the conditions \eqref{Invs} 
and \eqref{YBE} reduce to the following single inversion identity and single (\textit{standard}) YBE:
\begin{align}
    \raisebox{-1.05cm}{\RotMarkRedInvrFilteredregwOrGrSpecFramed}\; 
    = \raisebox{-1.05cm}{\FramedInvRHS}\ , \qquad\quad 
    \raisebox{-1cm}{\RotMarkRedYBEgoLeftSingleRregOrGrFramed}\; 
   =\raisebox{-1cm}{\RotMarkRedYBEgoRightSingleRregOrGrFramed}\ .
\end{align}
Relative to the involution $\cdot^*$ on the BMW algebra with $\tau,q\in\mathbb{R}$, the $R$-operator 
\eqref{equ:BMWRopParam} is self-adjoint for all $u\in\mathbb{R}\setminus\{q^2,\omega\}$.
\medskip

\noindent
\textbf{Remark.}
We have verified that, up to a normalising factor, the generalised Yang--Baxter framework of 
Proposition \ref{Prop:YBE} does not admit any other non-specious solution of the form (\ref{RegK}), 
with $r_g$ nonzero, than the one presented above.

\section{Liu algebra}
\label{Sec:Liu}

\subsection{Planar algebra}

The \textit{Liu planar algebra} is constructed as a quotient planar algebra much akin to the PSG planar 
algebra in Section \ref{Sec:YBIntYBR}. With $C^{(\epsilon)}(\alpha, \delta)$ as in \eqref{equ:PSGRels}, we thus 
introduce 
\begin{align}\label{equ:LiuRels}
    C_{\mathrm{L}}^{(\epsilon)}(\delta):= 
    C^{(\epsilon)}(0,\delta)\cup 
    \Big\{\!\scalebox{0.85}{\raisebox{-0.45cm}{\YBETangleLHSLiu}} 
    -\!\scalebox{0.85}{\raisebox{-0.45cm}{\YBETangleRHSLiu}}  
    -\frac{1}{\delta^2}\Big[\scalebox{0.85}{\raisebox{-0.425cm}{\TLvecSI}}\! 
    -\!\scalebox{0.85}{\raisebox{-0.425cm}{\TLvecSII}}\!
     - \epsilon\,\Big(\!\scalebox{0.85}{\raisebox{-0.425cm}{\ESLiuRelI}}\! 
     -\!\scalebox{0.85}{\raisebox{-0.425cm}{\SELiuRelI}}\!
     +\!\scalebox{0.85}{\raisebox{-0.425cm}{\ESLiuRel}}\! 
     -\!\scalebox{0.85}{\raisebox{-0.425cm}{\SELiuRel}} \!\Big)\Big]\Big\},
\end{align}
where $\delta\neq0$ and $\epsilon\in\{-\mathrm{i},\mathrm{i}\}$. Following \cite{Liu15}, the 
\textit{Liu planar algebra} $\mathrm{L}^{(\epsilon)}(\delta)$ is then defined as the quotient planar algebra 
$(A_n(S,C_{\mathrm{L}}^{(\epsilon)}(\delta)))_{n\in\mathbb{N}_0}$, where $S$ is as in \eqref{equ:PSGLabs}.
\medskip

\noindent
\textbf{Remark.}
In the PSG planar algebra in Section \ref{Sec:YBIntYBR}, for $\mathrm{PS}_1$ and $\mathrm{PS}_2$ to be 
positive-definite, we have $\delta>1$, see (\ref{d1}). Here, in our definition of the Liu planar algebra, we relax this 
condition to $\delta\neq0$.
\medskip

\noindent
From here onward, we opt for the abridged notation $L_n\equiv A_n(S,C_{\mathrm{L}}^{(\epsilon)}(\delta))$, 
$n\in\Nbb_0$. Imposing the relations $C_{\mathrm{L}}^{(\epsilon)}(\delta)$ tames the dimensionality of the 
universal planar algebra, resulting in the dimension formula
\be
 \dim L_n=(2n-1)!!.
\ee

The Liu planar algebra $\mathrm{L}^{(\epsilon)}(\delta)$ is spherical and involutive \cite{Liu15}, with the involution 
$\cdot^*$ defined as the conjugate linear map that acts by reflecting every disk about a line perpendicular to its 
marked boundary interval, as
\begin{align}
    \left(\!\!\raisebox{-0.285cm}{\BMWidIIAlt}\!\!\right)^* = \raisebox{-0.285cm}{\BMWidIIAlt}, \qquad\quad 
    \left(\!\!\raisebox{-0.285cm}{\BMWeProjAlt}\!\!\right)^* = \raisebox{-0.285cm}{\BMWeProjAlt}, \qquad\quad 
    \left(\!\!\raisebox{-0.285cm}{\SOpSG}\!\!\right)^* = \raisebox{-0.285cm}{\SOpSG},
\end{align}
recalling that $s^*=s$ for $\alpha=0$, see comment following (\ref{conj}).

We let $\Lc^{(\epsilon)}$ denote the set of all $\delta$ such that $\mathrm{L}^{(\epsilon)}(\delta)$ is positive 
semi-definite. For each $\delta\in\Lc^{(\epsilon)}$ and $\epsilon\in\{-\mathrm{i},\mathrm{i}\}$, 
the \textit{Liu subfactor planar algebra} $\mathsf{L}^{(\epsilon)}(\delta)$ is then defined as the quotient of 
$\mathrm{L}^{(\epsilon)}(\delta)$ by the kernel of the trace norm. Details of the set $\Lc^{(\epsilon)}$ are presented 
in \cite{Liu15}, including 
$\big\{\mathrm{i}\,\frac{q_m+q_m^{-1}}{q_m-q_m^{-1}}\,|\,m\in\mathbb{N}\big\}\subseteq\Lc^{(\epsilon)}$, 
where $q_m=q^{\frac{\mathrm{i}\pi}{2m+2}}$.

\subsection{Presentation and quotient description}
\label{Sec:Liupresentation}

For each $n\in\Nbb$, $\epsilon\in\{-\mathrm{i},\mathrm{i}\}$ and 
$\delta\neq0$, the \textit{Liu algebra} $\mathrm{L}_n^{(\epsilon)}(\delta)$ is defined by endowing the vector space 
$L_n$ with the multiplication induced by the unshaded planar tangle $M_n$ following 
from \eqref{equ:MultAndCoMultPlanar}, subject to the relations \eqref{equ:LiuRels}. 
We note that the Liu algebra is both unital (with unit denoted by 
$\mathds{1}$) and associative, and that it is a $^*$-algebra with involution inherited from the Liu planar algebra. 
The generators can be represented diagrammatically as in (\ref{1es}), and the algebra admits a presentation 
$\langle e_i,s_i\,|\,i=1,\ldots,n-1\rangle$ subject to the relations
\be
\begin{array}{rllll}
  & s_i^2=\mathds{1}-\tfrac{1}{\delta}e_i,\qquad 
  & e_is_i= s_ie_i = 0,\qquad
  & s_is_j=s_js_i,\qquad
  &  |i-j|>1,
\\[.15cm]
  & e_ie_{i\pm1}e_i= e_i,\qquad
  & e_ie_{i\pm1}s_i=\epsilon^{\pm1}e_is_{i\pm1},\qquad
  & s_ie_{i\pm1}e_i=\epsilon^{\mp1}s_{i\pm1}e_i,
  &
\end{array}
\label{Lrel}
\ee
and
\be
 s_{i}s_{i+1}s_{i} - s_{i+1}s_{i}s_{i+1} = \tfrac{1}{\delta^2}\big[s_i - s_{i+1} 
  - \epsilon\,(e_i s_{i+1} - s_{i+1}e_i + e_{i+1}s_{i} - s_{i}e_{i+1})\big],
  \label{LrelTrip}
\ee
with $\mathds{1}$ denoting the unit. Following from (\ref{Lrel}), we also have the relations
\be
 e_i^2=\delta e_i,\qquad 
 e_ie_j=e_je_i,\qquad 
 e_is_j=s_je_i,\qquad
 |i-j|>1,
\ee
\be
  e_is_{i\pm1}s_i=\epsilon^{\mp1}\big(e_ie_{i\pm1}-\tfrac{1}{\delta}\,e_i\big),\qquad
  s_is_{i\pm1}e_i=\epsilon^{\pm1}\big(e_{i\pm1}e_i-\tfrac{1}{\delta}\,e_i\big),
\ee
and
\be
  e_is_{i\pm1}e_i=0,\qquad s_ie_{i+1}s_i=s_{i+1}e_is_{i+1}.
  \label{LrelF}
\ee
We note that it suffices to list one of the two relations $e_is_i= 0$ and $s_ie_i = 0$ in (\ref{Lrel}). 

\medskip

\noindent
\textbf{Remark.}
It follows from the presentation above that the minimal vanishing polynomial in $s_i$ is $s_i^3-s_i$,
so $s_i$ is not invertible.
\medskip

The Liu algebra differs from the FC and BMW algebras in that there is no known basis for 
$\mathrm{L}_n^{(\epsilon)}(\delta)$ in terms of which the relations \eqref{Lrel}--\eqref{LrelF} admit 
a natural diagrammatic representation. In Section \ref{Sec:Braid}, we consider a basis which includes 
a braid \cite{Liu15}, and where \eqref{LrelTrip} can be interpreted as a type-III Reidemeister move. 
However, in this basis, some of the other relations fail to have a natural diagrammatic interpretation.

By comparing the presentation above with the one of $\mathrm{PS}_n^{(\epsilon)}(\alpha,\delta)$ 
in Section \ref{Sec:PSpresentation}, we obtain the following result, recalling that $\alpha=0$ 
for $\epsilon\in\{-\mathrm{i},\mathrm{i}\}$.
\begin{prop}\label{Prop:Liuq}
For each $n\in\Nbb$, $\epsilon\in\{-\mathrm{i},\mathrm{i}\}$, and $\delta>1$, we have
\be
 \mathrm{L}_n^{(\epsilon)}(\delta)\cong
  \mathrm{PS}_n^{(\epsilon)}(0,\delta)
  \big/\big\langle\iota_1,\ldots,\iota_{n-2}\big\rangle,
\ee
where
\be
 \iota_i=s_{i}s_{i+1}s_{i} - s_{i+1}s_{i}s_{i+1} -\tfrac{1}{\delta^2}\big[s_i - s_{i+1} 
 - \epsilon\,(e_i s_{i+1} - s_{i+1}e_i + e_{i+1}s_{i} - s_{i}e_{i+1})\big].
\ee
\end{prop}

We let $\Lc^{(\epsilon)}_n$ denote the set of all $\delta$ such that the trace form \eqref{form} is positive 
semi-definite on $L_{n}$, noting that $\Lc^{(\epsilon)}\subseteq\Lc^{(\epsilon)}_n$ for all $n\in\mathbb{N}_0$. 
For each $\delta\in\Lc^{(\epsilon)}_n$, $\mathsf{L}_n^{(\epsilon)}(\delta)$ is defined as 
the \textit{Liu subfactor algebra} constructed as the quotient of $\mathrm{L}_n^{(\epsilon)}(\delta)$ by the kernel 
of the trace norm.

\subsection{Baxterisation}

Relative to the canonical $L_2$-basis $\{\mathds{1}_2,e,s\}$, 
we introduce the parameterised $R$-operator as
\begin{align}\label{ResK}
 R(u) = r_{\mathds{1}}(u)\mathds{1}_2 + r_{e}(u)e + r_s(u)s,\qquad
 \raisebox{-0.325cm}{\RotRopGreenOpBBTLSTAlt} 
     =  r_{\mathds{1}}(u)\raisebox{-0.325cm}{\RotRopTLxIdSTAlt} 
     + r_{e}(u)\raisebox{-0.325cm}{\RotRopTLxESTAlt} 
     + r_s(u)\raisebox{-0.325cm}{\RotRopTLxBRSTAlt}\;,
\end{align}
with $r_{\mathds{1}},r_{e},r_s:\Omega\to\Cbb$.
\medskip

\noindent
\textbf{Remark.}
Unlike the $R$-operators in the FC and BMW algebras, the above $R$-operator is not invariant 
under rotation by $\pi$ (unless $r_s(u)=0$),
\begin{align}
    \raisebox{-0.325cm}{\RotRopGreenOpBBTLSTAltPiRot}\;
   \neq\raisebox{-0.325cm}{\RotRopGreenOpBBTLSTAlt},
\end{align}
since
\be
 \Pb_{r_{2,2}}[R(u)]=r_{\mathds{1}}(u)\mathds{1}_2 + r_{e}(u)e - r_s(u)s.
\label{Pr2R}
\ee
We also note that
\be
 \Pb_{r_{2,\pm1}}[R(u)]=r_{e}(u)\mathds{1}_2 + r_{\mathds{1}}(u)e +\epsilon^{\pm 1} r_s(u)s.
\label{Pr1R}
\ee

\noindent
\textbf{Remark.}
Since $e_i$ is quadratic in $s_i$, the function $r_s$ is required to be nonzero when exploring homogeneous 
integrability encoded by $\mathrm{L}_n^{(\epsilon)}(\delta)$, see Remark following \eqref{equ:RYOpSGBax}.
\begin{lem}\label{Lem:Liu}
Let
\be
 \raisebox{-0.345cm}{\RotRopGreenOpBBTLSTAltOrNOu}
   = \raisebox{-0.345cm}{\RotRopTLxIdSTAlt} 
   + \rho\raisebox{-0.345cm}{\RotRopTLxESTAlt} 
   + \sigma\raisebox{-0.345cm}{\RotRopTLxBRSTAlt},\qquad
  \raisebox{-0.345cm}{\RotRopDarkRedOpBBiNOi} 
  = \raisebox{-0.345cm}{\RotRopTLxIdSTAlt} 
   + \frac{\delta\rho(\delta+\rho)}{\sigma^2-\delta\rho}\raisebox{-0.345cm}{\RotRopTLxESTAlt} 
   -\frac{\delta\sigma(\delta+\rho)}{\sigma^2-\delta\rho}\raisebox{-0.345cm}{\RotRopTLxBRSTAlt},
\ee
where $\rho,\sigma\in\Cbb$ such that $\sigma^2\neq\delta\rho$. Then,
\be
 \raisebox{-1.05cm}{\OrangeTealFramed}
 \ =\raisebox{-1.05cm}{\TealOrangeFramed}
 \ = \frac{\delta(\delta+\rho)(\rho^2+\sigma^2)}{\sigma^2-\delta\rho}\raisebox{-1.05cm}{\FramedInvRHS}\,.
\ee
\end{lem}
\begin{proof}
The result follows by direct computation.
\end{proof}

\noindent
Using this result together with Proposition \ref{Prop:YBE} and Proposition \ref{Prop:YR}, as well as (\ref{Pr2R}) 
and (\ref{Pr1R}), we obtain the homogeneous Baxterisation of the Liu algebra in Proposition \ref{prop:BaxAlgL} 
below. To describe it, we find it useful to introduce the function
\be
 \phi:x\mapsto\mathrm{i}\,\frac{1+x}{1-x}
\ee
and the parameter
\be
 \Delta:=\frac{\mathrm{i}-\delta}{\mathrm{i}+\delta}=-\mathrm{i}\,\phi(\mathrm{i}\,\delta),
\ee
noting that $\phi\big(\frac{1}{x}\big)=-\phi(x)$.
\begin{prop}\label{prop:BaxAlgL}
For each $\mu\in\{-1,1\}$,
\be
 R(u)=
 \frac{1}{\phi(u)-\delta}
 \big[\phi(u)\mathds{1}_2+e+\mu\delta s\big]
\label{RLiu}
\ee
provides a homogeneous Baxterisation of $L_n^{(\epsilon)}(\delta)$, with
\be
\begin{array}{rll}
 &Y\!_1(u,v)=R(uv),\qquad &
 \overline{Y\!}_1(u,v)=\phi(uv)\phi\big(\tfrac{\Delta}{uv}\big)\Pb_{r_{2,2}}\big[R\big(\tfrac{\Delta}{uv}\big)\big],
 \\[.25cm]
 &Y\!_2(u,v)=\Pb_{r_{2,-1}}\big[R\big(\tfrac{u}{v}\big)\big],\qquad & 
 \overline{Y\!}_2(u,v)=\Pb_{r_{2,-1}}\big[R\big(\tfrac{v}{u}\big)\big],
 \\[.25cm]
 &Y\!_3(u,v)=\Pb_{r_{2,1}}\big[R\big(\tfrac{u}{v}\big)\big],\qquad & 
 \overline{Y\!}_3(u,v)=\Pb_{r_{2,1}}\big[R\big(\tfrac{v}{u}\big)\big].
\end{array}
\ee
\end{prop}
\noindent
\textbf{Remark.}
As expressed in (\ref{RLiu}), the $R$-operator appears ill-defined at $u=1$. However, following a simple
rewriting, it may be evaluated at $u=1$, yielding $R(1)=\mathds{1}_2$. The Baxterisation is thus well-defined for 
all $u\in\mathbb{C}$ for which $\phi(u)\neq\delta$, that is, for all $u\in\Cbb\setminus\{-\Delta\}$.
\medskip

Unlike the $R$-operators in the FC and BMW algebras, the $R$-operator \eqref{RLiu} is not crossing symmetric, 
that is, there do not exist scalar functions $\tilde{c}_R$ and $c_R$ such that
\begin{align}
    \Pb_{r_{2,1}}[R(u)] = \tilde{c}_R(u) R(c_R(u)).
\end{align}
This explains why the $Y$-operators are not expressible in terms of the $R$-operator itself, as is the
situation in the FC and BMW cases, c.f.~(\ref{Y1FC})--(\ref{equ:FCCrossSym}) 
and (\ref{Y1BMW})--(\ref{equ:BMWCrossSym}), respectively. In the Liu case, the conditions \eqref{Invs} and 
\eqref{YBE} reduce to the inversion identities
\begin{align}
    \raisebox{-1.05cm}{\RotMarkRedInvrFilteredregwOrGrSpecFramed}\; 
    = \raisebox{-1.05cm}{\FramedInvRHS} \,
    =\phi(w)\phi\big(\tfrac{\Delta}{w}\big)\raisebox{-1.05cm}{\FramedInvLHSLiu}\ , 
\end{align}
and a single (\textit{standard}) YBE 
\begin{align}
    \raisebox{-1cm}{\RotMarkRedYBEgoLeftSingleRregOrGrFramed}\; 
   =\raisebox{-1cm}{\RotMarkRedYBEgoRightSingleRregOrGrFramed}\ .
\end{align}

Relative to the involution $\cdot^*$ on the Liu algebra with $\delta\in\Rbb$, the $R$-operator \eqref{RLiu} 
is self-adjoint for all $u\in\mathbb{C}$ such that $|u|=1$ and $u\neq-\Delta$.
\medskip

\noindent
\textbf{Remark.}
We have verified that, up to a normalising factor, the generalised Yang--Baxter framework of 
Proposition \ref{Prop:YBE} does not admit any other non-specious solution of the form (\ref{ResK}), 
with $r_s$ nonzero, than the one presented in Proposition \ref{prop:BaxAlgL}.

\subsection{Braid limits}
\label{Sec:Braid}

For each $\mu\in\{-1,1\}$, the $R$-operator (\ref{RLiu}) yields well-defined $L_2$-elements 
under the specialisation $u=0$ and in the limit $|u|\to\infty$,
\be
 (1+\mathrm{i}\,\delta)R(0)=\mathds{1}_2-\mathrm{i}\,e-\mathrm{i}\,\mu\delta s, \qquad
 \lim_{|u|\to\infty}(1-\mathrm{i}\,\delta)R(u)=\mathds{1}_2+\mathrm{i}\,e+\mathrm{i}\,\mu\delta s,
\label{braidlimits}
\ee
and we collect the ensuing four elements in the set
\be
 \mathcal{H}=\{\mathds{1}_2+\epsilon_1e+\epsilon_2\delta s\,|\,\epsilon_1,\epsilon_2\in\{-\mathrm{i},\mathrm{i}\}\}.
\ee
For each $b\in\mathcal{H}$, we have 
\be
 b^2=2b-(\delta^2+1)\mathds{1}_2,
\ee
so, for $\delta\neq\pm\mathrm{i}$, $b$ is invertible, with inverse given by
\be
 b^{-1}=\tfrac{1}{\delta^2+1}(2\mathds{1}_2-b).
\ee
Although not obtained as limits of our $R$-operator, the elements of $\mathcal{H}$ also feature in \cite{Liu15},
where it is shown that for each $b\in\mathcal{H}$, the generators $\{b_1,\ldots,b_{n-1}\}\subset L_n$ satisfy
\be
 b_ib_{i\pm1}b_i= b_{i\pm1}b_ib_{i\pm1},\qquad
 b_ib_j= b_jb_i,\qquad  |i-j|>1,
\ee
and under a specialisation of $\delta$, this set generates an Iwahori-Hecke algebra. 
Here, $b_i\in L_n$ denotes the element `acting' as $b$ on the $i^{\mathrm{th}}$ and $(i+1)^{\mathrm{th}}$ nodes 
and as the identity elsewhere. This justifies referring to (\ref{braidlimits}) as \textit{braid limits}, noting that the 
particular limits and pre-factors in (\ref{braidlimits}) are a consequence of the normalisation choice in (\ref{RLiu}).

\section{Polynomial integrability}
\label{Sec:pol}

Following \cite{PR22}, let $\Ac$ be a finite-dimensional unital associative semisimple algebra, and 
suppose a model is described by a transfer operator $T(u)\in\Ac$, where $u\in\Omega$ with 
$\Omega\subseteq\Cbb$ a suitable domain. We then say that the model is \textit{polynomially integrable} 
(on $\Omega$) if there exists $b\in\Ac$ such that
\begin{align}\label{equ:PolInt}
    T(u)\in\Cbb[b], \qquad \forall\,u\in\Omega,
\end{align}
in which case $T(u)$ is said to be \textit{polynomialisable}.
Indeed, the integrability of such a model readily follows from \eqref{equ:PolInt}.
For the existence of $b$, we have the following corollary to results in \cite{PR22}, 
wherein an element of $\Ac$ is said to be \textit{diagonalisable} if the regular representation of the element is.
\begin{cor}\label{cor:Cxdiag}
Let $\Ac$ be a finite-dimensional unital associative semisimple algebra, and suppose that 
$\{C(u)\in\Ac\,|\,u\in\Omega\}$ is a one-parameter family of commuting and diagonalisable operators, with 
$\Omega\subseteq\Cbb$ a suitable domain. Then, there exists $b\in\Ac$ such that $C(u)\in\Cbb[b]$ for all 
$u\in\Omega$.
\end{cor}

\noindent
\textbf{Remark.}
The algebraic structure underlying the polynomial integrability of a model is, in a sense, 
as simple as possible, owing to its characterisation as a polynomial ring in a \textit{single} (parameter-independent) variable. 
We stress that this result applies in the familiar setting where $\mathcal{A}$ is a matrix algebra. 
In this case, polynomial integrability arises if the transfer matrix belongs to a commuting family 
of diagonalisable matrices -- a common situation for such models.
\medskip

Corollary \ref{cor:Cxdiag} indicates that YBR planar algebras offer prototypical examples of algebras 
underlying polynomially integrable models, as they 
(i) have a natural Yang--Baxter integrable structure, 
(ii) have an in-built inner product (the positive-definite trace form), and 
(iii) are semisimple. 
For such models, the polynomial integrability thus follows from the local properties 
\eqref{Invs}--\eqref{BYBEs} of Yang--Baxter integrability, provided the $R$- and $K$-operators \eqref{RKK} are 
self-adjoint, c.f.~Corollary \ref{cor:genDiagFaith}. Proposition \ref{prop:pol} below is a consequence of this.

\subsection{Singly generated algebras}
\label{Sec:Singly}

For each of the singly generated planar algebras FC, BMW, and Liu, we refer to the homogeneous transfer 
operator built using the $R$-operator parameterised in \eqref{equ:FCRopParam}, \eqref{equ:BMWRopParam}, 
and \eqref{RLiu}, respectively, as the \textit{canonical transfer operator} on $\Omega$, 
with $\Omega\subseteq\Cbb$ a suitable domain (that depends on the underlying algebra).
In each case, this transfer operator is the unique (up to renormalisations and reparameterisations) algebra 
element encoding homogeneous Yang--Baxter integrability. 

Using Liu's Theorem \ref{thm:LiuYBRClass}, the comments following Corollary \ref{cor:Cxdiag}, and results 
obtained in the previous three sections, we can now account for the \textit{polynomialisability} of the transfer 
operator in Theorem \ref{cor:1}. With notation as in Section \ref{Sec:FC}, Section \ref{Sec:BMW}, 
and Section \ref{Sec:Liu}, the following result thus gives conditions on the various algebra-defining 
and spectral parameters, ensuring that the respective canonical transfer operator is polynomialisable.
\begin{prop}\label{prop:pol}
\be
\begin{array}{rl}
(\mathrm{FC}):&
\text{Let $n\in\Nbb$ and $\gamma\in\mathcal{F}_{n}$, and suppose
$T_n(u)\in\mathsf{FC}_n(\gamma)$ is the corresponding canonical}
\\
 &\text{transfer operator, with $u\in\Rbb\setminus\{\gamma^2-1\}$. Then, $T_n(u)$ is polynomialisable.}
\\[.25cm]
(\mathrm{BMW}):&
\text{Let $n\in\Nbb$ and $(\tau,q)\in\mathcal{B}_n\cap\Rbb^2$, and suppose
$T_n(u)\in\mathsf{BMW}_n(\tau,q)$ is the corresponding}
\\
 &\text{canonical transfer operator, with $u\in\Rbb\setminus\{q^2,\omega\}$. Then, $T_n(u)$ is polynomialisable.}
\\[.25cm]
(\mathrm{Liu}):&
\text{Let $n\in\Nbb$, $\epsilon\in\{-\mathrm{i},\mathrm{i}\}$, and $\delta\in\Lc^{(\epsilon)}_{n}$, 
and suppose $T_n(u)\in\mathsf{L}_n^{(\epsilon)}(\delta)$ is the corresponding}
\\
 &\text{canonical transfer operator, with $|u|=1$ and $u\neq-\Delta$.
 Then, $T_n(u)$ is polynomialisable.}
\end{array}
\nonumber
\ee
\end{prop}

\noindent
\textbf{Remark.}
As the TL subfactor planar algebra is a planar subalgebra of every subfactor planar algebra, 
we have $\delta\in\mathcal{T}\subset\Rbb$ for any subfactor planar algebra, see (\ref{TTL}).
In the Liu case, in particular, it thus holds that $\Lc^{(\epsilon)}_{n}\cap\Rbb = \Lc^{(\epsilon)}_{n}$
for all $n\in\Nbb$ and $\epsilon\in\{-\mathrm{i},\mathrm{i}\}$.
\medskip

We are yet to determine a polynomial integrability generator (i.e.~$b$ in \eqref{equ:PolInt}) for any of the FC, BMW, and Liu models. 
In \cite{PR22}, we do it for the Temperley-Lieb model (see e.g.~\cite{PRZ06}), 
for small $n$, and find that the transfer 
operator is polynomialisable outside the domain specified by the corresponding version of Proposition \ref{prop:pol}. 
We anticipate similar results for the FC, BMW, and Liu models, that is, we expect that
the domains of polynomial integrability stated in Proposition \ref{prop:pol} are not maximal.

\section{Discussion} 
\label{Sec:Discussion}

Based on Liu's classification \cite{Liu15} and the algebraic integrability framework developed in \cite{PR22}, 
we have shown that the only singly generated (subfactor) planar algebras
that can encode homogeneous Yang--Baxter integrability are the YBR planar algebras. In the process,
we have provided a Baxterisation of the Liu planar algebra, and characterised the FC, BMW, and Liu
algebras as quotients of the proto-singly-generated algebra developed in \cite{BJ00,BJ03,BJL17,Liu15}.

In a forthcoming paper, we show how the present work offers a natural framework for defining and describing 
a Yang--Baxter integrable model with an underlying Iwahori-Hecke algebraic structure, see e.g.~\cite{Jones87} 
for an exposition of the Iwahori-Hecke algebra. 
Continuing \cite{PRT14}, we also intend to make precise how the $(2\times2)$-fused Temperley--Lieb algebra 
\cite{BR89,FR02,ZinnJustin07,MDPR14} arises as a quotient of the BMW algebra.

In lattice-model language, our transfer operators are constructed on the strip. When extending to the cylinder or 
annulus, transfer operators may be constructed from \textit{affine tangles}, in which case 
the operators are morphisms of the affine category of a given planar algebra \cite{Jones01,Ghosh11},
see also \cite{PS90,Levy91,MS94} on the so-called periodic Temperley--Lieb algebra.
We hope to return elsewhere with a discussion of how our findings extend to the annular setting; preliminary 
results can be found in \cite{Poncini23}.

A natural continuation of the present work would be to examine \textit{doubly generated} planar algebras,
for example as \textit{singly generated extensions} of the FC, BMW, or Liu algebras.
We would also find it interesting to study models with \textit{inhomogeneous} Yang--Baxter integrability, with or 
without shading, and to explore the physical properties of the various models, including integrals of motion and the 
continuum scaling limit.

\subsection*{Acknowledgements}

This work was supported by the Australian Research Council under the Discovery Project scheme, 
project number DP200102316.
The work of XP was also supported by an Australian Postgraduate Award from the Australian Government.
The authors thank Zhengwei Liu and Paul Pearce for discussions and comments.
They also thank the anonymous reviewer for their thoughtful comments and suggestions.

\appendix

\section{Planar algebras}
\label{app:pa}

\subsection{Naturality}
\label{app:Nat}

Planar tangles can be {\em glued} or composed as follows. Let $T$ and $S$ be planar tangles and suppose there 
exists $D\in\Dc_T$ satisfying $\eta(D)=\eta(D_0^S)$ and $\zeta(D)=\zeta(D_0^S)$.
It is then possible to isotopically deform $S$ such that it can take the place of $D$, as illustrated by
\begin{align}
 T =\raisebox{-1.175cm}{\CompTangTShaded}, \qquad\quad 
 S = \raisebox{-1.175cm}{\CompTangSShaded}, \qquad\quad 
 T\circ_D S = \raisebox{-1.43cm}{\CompTangTSShaded}.
\label{equ:CompositionTS}
\end{align}
Consistency between the composition of planar tangles and the action of the tangles as multilinear 
maps \eqref{PT} is often referred to as {\em naturality} and corresponds to 
\begin{align}\label{equ:NatSimpExII}
    \Pb_{T\circ_D S}(v_t, v_s)=\Pb_T(v_t, \Pb_S(v_s) ), \qquad
    v_t\in\!\bigtimes_{D_t\in\Dc_T\setminus\{D\}}\!A_{\eta(D_t)/2,\, \zeta(D_t)}, \quad 
    v_s\in\!\bigtimes_{D_s\in\Dc_S}\!A_{\eta(D_s)/2,\, \zeta(D_s)},
\end{align}
see \cite{JonesNotes, Poncini23} for more details.

\subsection{Unitality}
\label{app:Unit}

Let $(A_{n,\pm})_{n\in\mathbb{N}_0}$ be a shaded planar algebra.
\begin{lem}\label{lem:id_and_rot_invertible}
    If $A_{n,\pm}$ has no null vectors, then $\Pb_{\mathrm{id}_{n,\pm}}$ is the identity operator.
\end{lem}
\begin{proof}
Let $v\in A_{n,\pm}$ and $T$ be a planar tangle for which $\Pb_T$ has domain $A_{n,\pm}$. 
By naturality, we then have
\be
 \Pb_{T\circ_D\,\mathrm{id}_{n,\pm}}(v) = \Pb_{T}(\Pb_{\mathrm{id}_{n,\pm}}(v)),
\ee
hence
\be
   \Pb_{T}(v - \Pb_{\mathrm{id}_{n,\pm}}(v)) = 0,
\ee
so $v - \Pb_{\mathrm{id}_{n,\pm}}(v)\in \mathrm{ker}(\Pb_T)$.
Since $A_{n,\pm}$ has no null vectors, it follows that $\Pb_{\mathrm{id}_{n,\pm}}(v) = v$ for all $v\in A_{n,\pm}$.
\end{proof}
\noindent
Together with naturality, Lemma \ref{lem:id_and_rot_invertible} implies the following result.
\begin{cor}\label{cor:ElRot}
 If $A_{n,+}$ and $A_{n,-}$ have no null vectors, 
 then $\Pb_{r_{n,-\ell,\mp}}\circ\Pb_{r_{n,\ell,\pm}}=\Pb_{\mathrm{id}_{n,\pm}}$ 
 for each $\ell\in\{-1,1\}$.
\end{cor}
\begin{prop}\label{prop:AssocUnitAlg}
Let $A_{n,\pm}$ be endowed with the multiplication induced by $M_{n,\pm}$, and suppose
$A_{n,\pm}$ has no null vectors. Then, $A_{n,\pm}$ is unital, with unit $\mathds{1}_{n,\pm}$.
\end{prop}
\begin{proof}
    Let $D_l$ and $D_u$ denote the lower respectively upper disk in the planar tangle $M_{n,\pm}$, 
    and let $v\in A_{n,\pm}$. Then, naturality implies that
    \begin{align}
  v\mathds{1}_{n,\pm}
  =\Pb_{M_{n,\pm}}(v, \mathds{1}_{n,\pm})
  =\Pb_{M_{n,\pm}\circ_{D_u} \mathrm{Id}_{n,\pm}}(v)
  =\Pb_{\mathrm{id}_{n,\pm}}(v)
  =\Pb_{M_{n,\pm}\circ_{D_l} \mathrm{Id}_{n,\pm}}(v)
  =\Pb_{M_{n,\pm}}(\mathds{1}_{n,\pm},v)
  =\mathds{1}_{n,\pm}v.
    \end{align}
By Lemma \ref{lem:id_and_rot_invertible}, $\Pb_{\mathrm{id}_{n,\pm}}(v)=v$, 
hence $v\mathds{1}_{n,\pm}=v=\mathds{1}_{n,\pm}v$.
\end{proof}

\section{Temperley--Lieb algebra}
\label{app:TL}

The Temperley--Lieb planar algebra $(\mathrm{T}_{n,\pm})_{n\in\mathbb{N}_0}$ discussed in 
Section \ref{subsec:TLPA} admits an \textit{unshaded} version: the \textit{(unshaded) TL planar algebra} 
$(\mathrm{T}_{n})_{n\in\mathbb{N}_0}$, where the shading of planar tangles and disks in 
$(\mathrm{T}_{n,\pm})_{n\in\mathbb{N}_0}$ is omitted. For each $n\in\mathbb{N}$, the corresponding 
\textit{TL algebra} $\mathrm{TL}_{n}(\delta)$ is defined by endowing the vector space $\mathrm{T}_n$ with the 
multiplication induced by the unshaded companion to \eqref{equ:MultAndCoMultPlanar}. 
We note that $\mathrm{TL}_{n}(\delta)$ is both unital (with unit denoted by $\mathds{1}$) and associative, and that 
it is a $^*$-algebra with involution inherited from $(\mathrm{T}_{n})_{n\in\mathbb{N}_0}$. As is well-known 
\cite{Kauffman87}, the generators of $\mathrm{TL}_{n}(\delta)$ can be expressed diagrammatically as
\begin{align}
    \mathds{1}\ \leftrightarrow \raisebox{-0.885cm}{\BMWRecIdOdd},\qquad 
    e_i\ \leftrightarrow\raisebox{-0.885cm}{\BMWRecEOdd}\qquad (i=1,\ldots,n-1).
\end{align}

The TL algebra $\mathrm{TL}_{n}(\delta)$ admits a presentation:
\begin{align}
    \mathrm{TL}_n(\delta) \cong \langle e_1,\ldots,e_{n-1}\rangle,
\end{align}
subject to the relations
\begin{align}
        e_i^2 &= \delta e_i,\\
        e_je_{i}e_j &= e_{j}, \quad\quad\quad |i-j|=1,\\
        e_ie_j &= e_je_i, \quad\quad\: |i-j|>1.
\end{align}
For each $n\in\Nbb$, there exists unique $\wj_n\in\mathrm{TL}_{n}(\delta)$ such that
\be
 \wj_n^2=\wj_n,\qquad e_i\wj_n=\wj_ne_i=0,\qquad i=1,\ldots,n-1.
\label{wj}
\ee
These {\em Jones--Wenzl idempotents} \cite{Wenzl87} are used in Section \ref{Sec:Proto}.

\section{Braid-semigroup algebra quotient}
\label{app:BS}

For each $n\in\Nbb$, let the \textit{braid-semigroup algebra} $\mathrm{BS}_n$ be the associative unital algebra 
$\langle b_1,\ldots,b_{n-1}\rangle$ subject to
\begin{align}
 b_ib_jb_i&= b_jb_{i}b_j,\qquad  |i-j|=1,
\label{bbb}
 \\[.15cm]
 b_ib_j&= b_jb_i,\qquad \ \ \, |i-j|>1,
\label{bb}
\end{align}
with unit denoted by $\mathds{1}$. Since $\{b_1^k\,|\,k\in\Nbb\}$ is linearly independent, 
$\mathrm{BS}_n$ is infinite-dimensional for $n>1$. This may be remedied by quotienting out the ideal 
generated by an appropriate set of polynomials. In the following, we demonstrate how the BMW algebra 
$\mathrm{BMW}_n(\tau,q)$ can be obtained as such a quotient. Although this is likely known to experts, we have 
not been able to find it in the literature and include it for the interested reader.

Let $n\in\Nbb$, $\tau,Q\in\Cbb$ and
\begin{align}\label{equ:IotaI}
 \iota_i&:=(b_i^2-Qb_i-\mathds{1})(\tau b_i-\mathds{1}),
\\[.15cm]
 \iota_{i,j}^{(1)}&:=Q(b_ib_ib_jb_i-Qb_ib_jb_i-b_jb_i)+\tau(b_i^2-Qb_i-\mathds{1})(b_j^2-Qb_j-\mathds{1}), 
\\[.15cm]
 \iota_{i,j}^{(2)}&:=b_jb_i(b_j^2-Qb_j-\mathds{1})b_ib_j-(b_i^2-Qb_i-\mathds{1}),
\end{align}
for all $i,j\in\{1,\ldots,n-1\}$.
\begin{prop}
With $Q=q-q^{-1}$, we have
\be
 \mathrm{BMW}_n(\tau,q)
 \cong\mathrm{BS}_n/\big\langle \iota_i, \iota_{i,j}^{(1)}, \iota_{i,j}^{(2)}\,|\,|i-j|=1;\,i,j\in\{1,\ldots,n-1\}\big\rangle.
\ee
\end{prop}
\begin{proof}
In the quotient algebra, $b_i$ is invertible, with inverse given by
\be
 b_i^{-1}=-\tau b_i^2+(\tau Q+1)b_i+(\tau-Q)\mathds{1},
\ee
even if $\tau=0$ or $Q=0$.
For $\tau,Q\neq0$, it is convenient to describe the quotient algebra using
\be
 e_i:=-\frac{\tau}{Q}(b_i^2-Qb_i-\mathds{1}),
\ee
noting that
\be
 b_i-b_i^{-1}=Q(\mathds{1}-e_i)
\label{ggm1}
\ee
and that
\be
 b_ie_j=e_jb_i,\qquad e_ie_j=e_je_i,\qquad |i-j|>1.
\ee
The vanishing of $\iota_i$ and $\iota_{i,j}^{(\ell)}$, $\ell=1,2$, now corresponds to
\be
 b_ie_i = e_ib_i = \tau^{-1}e_i
\ee
and
\be
 e_ib_{i\pm1}b_i= b_{i\pm1}b_ie_{i\pm1}= e_{i}e_{i\pm1},\qquad 
 b_{i}e_{i\pm1}b_{i} = b_{i\pm1}^{-1}e_{i}b_{i\pm1}^{-1},
\label{iotaij}
\ee
respectively, where the first equality in (\ref{iotaij}) is a simple consequence of the braid relations (\ref{bbb}).
It also follows that
\be
 b_{i}e_{i\pm 1}e_{i} = b_{i\pm 1}^{-1}e_{i},\qquad
    e_{i}e_{i\pm 1}b_{i} = e_{i}b_{i\pm 1}^{-1},\qquad
    e_{i}b_{i\pm 1}e_{i} = \tau e_{i},
\ee
and
\be
 e_i^2=\delta e_i,\qquad e_ie_{i\pm1}e_i=e_i,
\label{ee}
\ee
where $\delta$ is given as in (\ref{BMWdelta}).
\end{proof}

We note that the BMW algebra can similarly be expressed as a quotient of the braid-group algebra
but find the semigroup description above more appealing. In particular, the invertibility of the basic generators 
in $\mathrm{BMW}_n$ follows from \eqref{equ:IotaI}, while it is `built-in' if starting with the braid-group algebra.

\end{document}